\documentclass[a4paper,12pt]{article}
\usepackage[utf8]{inputenc}
\usepackage{adjustbox,caption,subcaption,setspace}
\usepackage[margin=1.25in]{geometry}
\usepackage{tikz,pgfplots}
\pgfplotsset{compat=1.18}
\usepackage{enumitem,bbm}
\usepackage{amssymb,amsmath,amsthm,mathtools}
\usepackage{newpxtext,newpxmath}
\usepackage[authoryear]{natbib}
\usepackage{xcolor}
\usepackage{hyperref}
\usepackage{todonotes}

\hypersetup{colorlinks, citecolor=gray!80!black}

\DeclarePairedDelimiter\ceil{\lceil}{\rceil}
\DeclareMathOperator*{\argmax}{arg\,max}
\DeclareMathOperator{\interior}{int}

\newcommand{\defn}[1]{\textcolor{red!75!black}{\emph{#1}}}

\newcommand\blfootnote[1]{
  \begingroup
  \renewcommand\thefootnote{}\footnote{#1}
  \addtocounter{footnote}{-1}
  \endgroup
}

\newtheorem{innercustomgeneric}{\customgenericname}
\providecommand{\customgenericname}{}
\newcommand{\newcustomtheorem}[2]{%
  \newenvironment{#1}[1]
  {%
   \renewcommand\customgenericname{#2}%
   \renewcommand\theinnercustomgeneric{##1}%
   \innercustomgeneric
  }
  {\endinnercustomgeneric}
}

\newtheorem{theorem}{Theorem}
\newtheorem{proposition}{Proposition}
\newtheorem{lemma}{Lemma}
\newtheorem{assumption}{Assumption}

\theoremstyle{definition}

\newtheorem*{notation*}{Notation}
\newtheorem{remark}{Remark}

\newcustomtheorem{customassumption}{Assumption}

\title{Economic dynamics with differential fertility}
\author{Francis Dennig\thanks{United Nations
Development Programme, Rome Centre; \textsf{francis.dennig@undp.org}}
    \quad Bassel Tarbush\thanks{Department of Economics, University of Oxford; \textsf{bassel.tarbush@economics.ox.ac.uk}}}
\date{\today}

\advance\textwidth2\evensidemargin
\begin{document}
\maketitle

\onehalfspacing
\begin{abstract}
\noindent We characterize the outcomes of a canonical deterministic model for the intergenerational transmission of capital that features differential fertility. A fertility function determines the relationship between parental capital and the number of children, and a transmission function determines the relationship between the capital of a parent and that of their children. Together these functions generate an evolving cross-sectional distribution of capital. We establish easy-to-verify conditions on the fertility and transmission functions that guarantee (a) that the dynamical system has a steady state distribution that is either atomless (exhibiting inequality) or degenerate (not exhibiting inequality), and (b) that the system converges to such states from essentially any initial distribution. Our characterization provides new insights into the link between differential fertility and long-run cross-sectional inequality, and it gives rise to novel comparative statics relating the two. We apply our results to several parametric examples and to a model of economic growth that features endogenous differential fertility.
\blfootnote{The findings presented in this paper are those of the authors and they do not necessarily represent the views of the UNDP or its affiliated organizations.}
\end{abstract}

\section{Introduction}
What is the theoretical relationship between differential fertility and long-run cross-sectional inequality? The equation of motion at the heart of many models of economic growth features \emph{differential fertility}; that is, people with different levels of wealth can have systematically different numbers of children \citep*{lam86,chu90,morand99, delacroix03,cordoba16,cavalcanti2021family}. However, existing analytic characterizations of the long-run distributional outcomes in such models have been relatively limited. The key challenge is that differential fertility complicates intergenerational dynamics so that standard tools for the analysis of dynamic systems do not always straightforwardly extend to the deterministic models often favored by the family economics literature.

In this paper, we provide an extensive analytical characterization of the outcomes of a canonical deterministic model for the intergenerational transmission of capital that features differential fertility, and we provide new insights regarding the general implications of differential fertility on the evolution of the cross-sectional distribution of capital across generations. The primitives of our system are a \emph{fertility function}, which describes the relationship between parental capital and fertility, and a \emph{transmission function}, which describes the relationship between parental capital and the capital level of the children. These two relationships are the subject of much empirical research, and they are also well-studied in the theoretical family economics literature (e.g.\ see \citealp*{jones08}). Together, the fertility and transmission functions generate our main object of interest: an evolving cross-sectional distribution of capital.

A key feature of our model is that it is entirely deterministic and yet we show that it can generate atomless steady state distributions of capital (that exhibit cross-sectional inequality) to which the system converges from essentially any initial distribution of capital. Some models of the intergenerational transmission of capital can yield long-run distributions of capital that are atomless but, in all cases that we are aware of, these are obtained only from the introduction of exogenous stochastic shocks  \citep*{delacroix03,benhabib14,jones15,piketty15,gabaix15}. The long-run atomless distributions of capital that our model generates arise from a different mechanism, and we can attribute any long-run distributional outcomes to the shapes of the fertility and transmission functions alone rather than to the dynamic implications of exogenous stochastic shocks. As we will argue, understanding the deterministic dynamics of our system can provide valuable insights into how capital evolves across generations even if one believes, as we do, that random shocks play an important role in integenerational transmission.

This is how we proceed:

\begin{enumerate}[topsep=0ex,itemsep=0ex,leftmargin=0cm,label=(\roman*)]
\item In the first instance, we assume that the primitives (the fertility and transmission functions) are exogenously given. We show that the dynamic system can be decomposed into multiple parallel `local' dynamic systems and that the solutions to these local systems can be aggregated back into a solution of the whole system. With this approach, we can characterize the long-run outcomes of our dynamic system. Some of our results will be familiar because they recover outcomes in models of poverty traps and in growth models with multiple equilibria \citep*{Galor1993Income,Azariadis2005Chapter,moav05}. Others, however, are new. In particular, we establish easy-to-verify conditions on the primitives that guarantee (a) that the dynamical system has a steady state distribution that is either atomless (exhibiting inequality) or degenerate (not exhibiting inequality), and (b) that the system converges to such states from essentially any initial distribution. To our knowledge our results offer the most complete theoretical characterization of the outcomes of this canonical system, particularly with regards to the atomless steady states.\footnote{Intergenerational models in the family economics and growth literature that have accounted for differential fertility tend to track the long-run outcomes coarsely; for example, \cite*{becker79} track the coefficient of variation---but not the distribution---of income while the models of \cite*{dahan98}, \cite*{kremer02}, and \cite*{fan2013differential} study the relationship between differential fertility and inequality, with inequality being measured by the ratio of skilled to unskilled workers (and their relative wages). By contrast, the full steady state distribution of income or of wealth (or its tail) can be derived explicitly in several stochastic models of the dynamics of economic inequality; for example, see \cite*{benhabib11,benhabib14,jones15,piketty15,gabaix15}. Inequality in such models is typically generated by an underlying (usually Kesten-type) stochastic process that results in a steady state belonging to a specific parametric family (often to the Pareto or lognormal families). However, this literature tends to ignore fertility differentials. The surveys of \cite*{davies00}, \cite*{gabaix09}, and \cite*{benhabib16}, for example, do not mention the role of fertility.} In the appendix, we provide a technical discussion of how our model relates to some prior work in mathematics regarding the (deterministic) analysis of branching systems.

Point (i) is the main contribution of the paper. Points (ii)-(iv) below are applications of our main characterization result.

\item For specific functional forms of our fertility and transmission functions, we produce explicit analytical solutions for the atomless steady state distributions. As we show, these can belong to known parametric families (e.g. Pareto).

\item  Our results make the dependence of the evolution of the distribution of capital on the primitives explicit which allows us to obtain novel comparative statics relating changes in the primitives to changes in the distributional outcomes of the dynamic process. In particular, we show that monotone likelihood shifts in fertility---which can be expressed as changes in the \emph{capital elasticity of fertility}---result in first-order stochastic dominance shifts in the steady state distribution of capital. This is a result that extends existing work. \cite*{lam86} and \cite*{chu90}, for example, work in a discrete income space and a stochastic transmission function (though the analysis of their model is done entirely via a deterministic Markov kernel), while we work in a continuous space and a deterministic transmission function, but in all other respects we study identical dynamic systems.\footnote{The underlying dynamic system in \cite*{lam86} and \cite*{chu90} is essentially a discretized version of \cite*{loury81} that accounts for fertility differentials.}  In their seminal contribution, \cite*{chu90} show that with a downward sloping fertility curve, and subject to some technical conditions, fertility reductions in the lowest income group lead to first order stochastic dominance shifts in the steady state income distribution. Their result is based on a non-constructive existence result for the steady state (as the solution to an eignevalue equation), which limits the scope for fine-grained comparative statics. Our result is more general in the sense that we can conclude a stochastic dominance shift in the distribution of capital from a change \emph{anywhere} in the fertility function, rather than only from a change in the fertility of the lowest income group. Moreover, our result does not rely on any particular shape for the fertility function and can thus accommodate downward sloping but also hump-shaped functions (e.g.\ as in \citealp*{vogl16b}).\footnote{People with lower income or wealth tend to have a larger number of children in modern societies \citep*{mace08,skirbekk08,balbo13}. Fertility may have been positively related to status or wealth before the demographic transition, though results are mixed \citep*{dribe14}.}  

\item While our main result on the existence of atomless steady states and convergence to these states is derived under the assumption that the primitives are exogenously given, we show that our results are also applicable in situations in which the fertility and transmission functions are endogenously determined by optimizing choices of parents in every generation. To do so, in Section \ref{sec:endo}, we analyze a model in which parents in each generation face a Becker-style quantity-quality trade-off when choosing the number of children that they have, with educational costs operating as an important constraint. Endogenously, poorer parents have more children and invest less in education per child. The fertility and transmission functions therefore arise endogenously yet we apply our results to easily derive the long-run distributional outcomes of this model.
\end{enumerate}

\section{Model} \label{sec:model}

People born in period $t$ each have a number of children who are born in period $t+1$, creating a succession of generations indexed by $t$. Procreation is assumed to be asexual.\footnote{Unlike the models of \cite*{blinder73} and \cite*{straub84}, there is no marriage in our model, and reproduction is asexual. We therefore ignore assortative matching as a determinant of inequality \citep*{bruze15,greenwood16,mare16}.} We denote by $x\in\mathcal{X}$ the \defn{capital} of an individual, and we  assume throughout that $\mathcal{X}\subseteq \mathbb{R}$ is a real interval. 

Depending on the application, `capital' may refer to wealth or lifetime income (physical capital), human capital, or social capital. At this point we don't impose a particular interpretation on capital and simply think of it as a variable with the property that the capital of a child is associated with that of its parent, and we capture this association with a parent-to-child capital \defn{transmission function} $\tau:\mathcal{X}\rightarrow\mathcal{X}$. Each child of a parent whose capital level is $x$ has capital level $\tau(x)$.\footnote{In this formulation, all children are treated in the same way by their parent and we therefore ignore unequal intra-family transmission processes such as primogeniture in inheritance, in which all of the inheritable physical capital is passed to the first-born child \citep*{stiglitz69,vaughan79,atkinson80}.}

Denote the number of children born to a parent with capital $x$ by $n(x)$. We refer to $n : \mathcal{X} \rightarrow (0,\infty)$ as the \defn{fertility function}.

The cross-sectional distribution of capital among individuals in generation $t$ is denoted by $F_t : \mathbb{R} \rightarrow [0,1]$. Individuals with capital $x$ have $n(x)$ children, so the average number of children born in generation $t+1$ to parents in generation $t$ is given by
\[
	\mathbb{E}_t[n] := \int_\mathcal{X} n(x)dF_t(x).
\]
Since we assume asexual reproduction, the average number of children $\mathbb{E}_t[n]$ is also the population growth factor.\footnote{In other words, letting $P_t$ denote the size of the population in generation $t$, we have $P_{t+1} = P_t \mathbb{E}_t[n]$.}

Given the fertility and transmission functions, the \defn{population process} is a mapping $F_t \mapsto F_{t+1}$ from the distribution at $t$ to the distribution at $t+1$ which, for each $x \in \mathcal{X}$, is given by:
\begin{equation}\label{main_nonlin}
	F_{t+1}(x) = \frac{1}{\mathbb{E}_t[n]}\int_{\mathcal{X}} \mathbbm{1}\left[\tau(z) \leq x\right]  n(z) dF_{t}(z) .
\end{equation}
This is our main equation of interest. It shows that the fraction of people in generation $t+1$ with capital no greater than $x$ is precisely the fraction of all children born in generation $t+1$ whose parental capital $z$ satisfies $\tau(z) \le x$.

Our aim is to study the dynamics of the sequence of distributions $( F_t )_{t \geq 0}$ generated by \eqref{main_nonlin}. A distribution $F^*$ is a \defn{steady state} of the population process if it is a fixed point of the map $F_t \mapsto F_{t+1}$. We say that the population process \defn{converges} to a steady state $F^* $ if
$$\lim_{t \rightarrow \infty} F_t = F^*.$$
We denote the population growth factor in such a steady state by 
\[
\mathbb{E}^*[n] = \int_{\mathcal{X}} n(x) dF^*(x).
\]
\begin{remark}\label{rem:scale}
The population process \eqref{main_nonlin} would induce the same sequence of distributions $(F_t)_{t\geq 0}$ if the fertility function $n(x)$ were replaced by $\tilde{n}_t(x) = C_t n(x)$ where $C_t > 0$ for each $t$. In other words, including a possibly time-varying scalar for fertility in no way affects the distributional dynamics: the scalar would cancel out with the denominator in \eqref{main_nonlin}. Only the population growth factor would be affected by the inclusion of such a scalar.
\end{remark}

\subsection{Discussion of the model}
The population process described by \eqref{main_nonlin} is canonical. It is a branching process in which each person in generation $t$ with capital $x$ has $n(x)$ children each with capital level $\tau(x)$ in generation $t+1$, and has $n(x) \times n(\tau(x))$ grandchildren each with capital level $\tau(\tau(x))$ in generation $t+2$, and so on. This process is, for example, studied in discretized form in \cite*{lam86} and in \cite*{chu90} but the properties of the mapping have not been fully characterized analytically. The same dynamic is also at the heart of more recent macroeconomic models that have accounted for differential fertility and in which some of the implications for inequality are assessed via calibration \citep*{delacroix03,cordoba16,cavalcanti2021family}.

Neither fertility nor the transmission function is a random variable, so one's fertility and the capital level of one's children are deterministic functions of one's own capital. We view the fact that our system is deterministic as a feature, not a bug. While stochastic components play an important role in intergenerational dynamics,\footnote{For example, see \cite*{davies00,gabaix09,benhabib16}.} our deterministic approach allows us to attribute any resulting inequality in the distribution of capital to the economic content captured by the fertility and transmission functions, and not to the dynamic implications of a stochastic error term. The inclusion of such a term would undoubtedly affect long-run inequality, but our approach  completely disentangles the fertility and transmission-related determinants of inequality from those related to random shocks; the latter are tied to social mobility, which is related to but distinct from cross-sectional inequality.\footnote{Consider a society in which, in each generation, each person's capital is drawn at random from some distribution $F^*$, and another society in which the distribution of capital is $F^*$ and each person has a single child to which they bequeath all of their capital. With large populations, these two societies have essentially the same distribution of capital from one generation to the next (and therefore have the same cross-sectional inequality), but the first society clearly exhibits a lot of social mobility whereas the second exhibits none.} More abstractly, it is entirely natural to ask how a system that is subject to random shocks would behave in the absence of these shocks, since this would clarify the role of the shocks. Moreover, as we show in Section \ref{sec:distro}, the results that we derive for our deterministic system have implications for empirical calibrations of intergenerational models.

The fertility and transmission functions in \eqref{main_nonlin} are exogenous and time-independent, and we maintain this restriction in Sections \ref{sec:model} through \ref{sec:results}. However, with appropriate re-normalization, our results are applicable in situations with time-dependent and endogenous fertility and transmission. We present such an application in Section \ref{sec:endo}, where we employ our results for \eqref{main_nonlin} to solve the dynamics of a growth model that features fertility and educational choices. There, fertility and transmission are endogenously determined in equilibrium by parental preferences and technological constraints and are therefore dependent of state variables and thus time.

\section{Main results} \label{sec:results}
We start with assumptions about properties of the fertility and transmission functions. 

\begin{assumption}\label{as1}
The fertility function $n: \mathcal{X} \rightarrow (0, \infty)$ is  positive, bounded, and continuously differentiable.\footnote{For the avoidance of doubt, we refer to $x \in \mathbb{R}$ as positive if $x>0$ and non-negative if $x \geq 0$.}
\end{assumption}
\noindent  A flat fertility function of the form $n(x)=C$ for some constant $C>0$ satisfies Assumption \ref{as1}; we therefore nest dynamics with non-differential fertility. One might want $n$ to be integer-valued (i.e. a step function) to better capture indivisibility, but (i) it is common in the literature to allow the fertility function to take non-integer values; a justification for this is to interpret $n(x)$ as being the average fertility of people with capital level $x$, and (ii) a fertility function can be made to approximate a step function arbitrarily closely while satisfying Assumption \ref{as1}.

We say that a fixed point $x$ of the transmission function $\tau$ is called a \defn{source} if $\tau'(x)>1$, and a \defn{sink} if $\tau'(x)<1$.
\begin{assumption}\label{as2}
The transmission function $\tau: \mathcal{X} \rightarrow \mathcal{X}$ satisfies the following conditions:
\begin{enumerate}[leftmargin=0.7cm, label=(\roman*)]
    \item $\tau$ is a twice continuously differentiable bijection whose first derivative is positive everywhere.
    \item $\tau$ has exactly $1 \leq K < \infty$ fixed points in $\mathcal{X}$, denoted by $s_1,\dots,s_K$ with $s_1 < \cdots < s_K$.
    \item Each fixed point of $\tau$ is either a source or a sink.
\end{enumerate} 
\end{assumption}
\noindent A transmission function satisfying Assumption \ref{as2} can have any finite number of fixed points. The third part of the assumption rules out non-generic transmission functions that are tangential somewhere along the 45 degree line. Observe that for any transmission function satisfying Assumption \ref{as2}, we can break up $\mathcal{X}$ into the following intervals: $\mathcal{X}_0:=(s_0,s_1]$ and $\mathcal{X}_K:=[s_K,s_{K+1})$, where $s_0:= \inf \mathcal{X}$ and  $s_{K+1}:=\sup\mathcal{X}$, and if $K>1$ then for each $k \in \{1,...,K-1\}$ let
\[
\mathcal{X}_k := [s_k,s_{k+1}] .
\]
\noindent We re-use the intervals $\mathcal{X}_0,\dots,\mathcal{X}_K$ throughout the text. Assumption \ref{as2} implies that $s_{K+1}$ is either infinite in which case $\mathcal{X}_K = [s_K,\infty)$, or it is finite, so $s_{K+1}=s_K$, in which case $\mathcal{X}_K=\emptyset$. Similarly, $\mathcal{X}_0$ is either $(-\infty,s_1]$ or empty.

We will maintain Assumptions \ref{as1} and \ref{as2} throughout Section \ref{sec:results} so we do not repeat the assumptions in the statement of each result in this section. 

\paragraph{Notation} Consider a real interval $I \subseteq \mathbb{R}$. We denote the interior of $I$ by $\interior I$. For any function $\mu : I \rightarrow I$ and integer $t \geq 1$,
\[
\mu^{[t]} := \underbrace{\mu \circ \mu \circ \cdots \circ \mu}_{t\text{-times}}
\]
denotes the $t$-fold composition of $\mu$, and we adopt the convention $\mu^{[0]}(x) :=x$. Letting $a$ and $b$ denote the endpoints of $I$, we say that a function $\mu:I \rightarrow (0,\infty)$ is \defn{endpoint maximal} at $a$ in $I$ if $\lim_{x \to a}\mu(x)>\lim_{x \to b} \mu(x)$. Finally, as is standard, we say that $\mu:\mathbb{R} \rightarrow [0,\infty)$ is supported on $I \subseteq \mathbb{R}$ if $I$ is the closure of the subset of $\mathbb{R}$ for which $\mu$ is non-zero.

\subsection{Degenerate initial distribution}\label{sec:atoms}
We start by showing that a distribution that is degenerate at a fixed point of the transmission function is always a steady state of the population process. Moreover, we establish conditions under which an initial distribution that is degenerate at a single level of capital converges to a degenerate steady state. The results of Section \ref{sec:atoms} are straightforward. We briefly discuss them here only for completeness.

\begin{remark}
Consider some $k \in \{0,\dots,K\}$ and a fixed point of the transmission function $s \in \mathcal{X}_k$. A distribution that is degenerate at $s$ is a steady state of the population process, i.e. $F^*(x) := \mathbbm{1}[s \leq x]$, and the  population growth factor in this steady state is given by $\mathbb{E}^*[n] = n(s)$. Additionally, if the initial distribution $F_0$ is degenerate at some interior point of $\mathcal{X}_k$ and $s$ is a sink then, regardless of the shape of the fertility function, the population process converges to a distribution that is degenerate at $s$. 
\end{remark}
The first part of the above remark is evident: since $s$ is a fixed point of $\tau$, anyone born with capital level $s$ will have descendants with capital level $\tau(s)=s$, so a distribution of capital $F^*$ that is degenerate at $s$ is a steady state of the population process. Moreover, the number of children that each parent has at $s$ is $n(s)$; this is therefore the population growth factor. The second part of the remark above states that if the initial distribution is $F_0(x) = \mathbbm{1}[y \leq x]$ for some $y \in \interior \mathcal{X}_k$ and $s \in \mathcal{X}_k$ is a sink, then $\lim_{t \to \infty} F_t(x) = \mathbbm{1}[s \leq x]$. The reason is, again, straightforward: since $F_{0}$ is degenerate at some $y \in \interior  \mathcal{X}_k$, every parent in generation $0$ has capital level $y$. This implies that each child in generation $1$ has capital level $\tau(y)$. In other words, $F_{1}(x) = \mathbbm{1}[\tau(y) \leq x] $. Carrying this forward to generation $t$ gives us
\[
F_{t}(x) = \mathbbm{1}[\tau^{[t]}(y) \leq x] .
\]
The result then follows because, since $s$ is a sink, $\lim_{t \to \infty} \tau^{[t]}(y) = s$ for each $y\in \interior \mathcal{X}_k$.\footnote{We provide a simple formal proof of this in the appendix; see Lemma \ref{lem:kappa1}.} An example is shown in Figure \ref{fig:atoms}.

\begin{figure}
    \centering
    \includegraphics[width=0.5\linewidth]{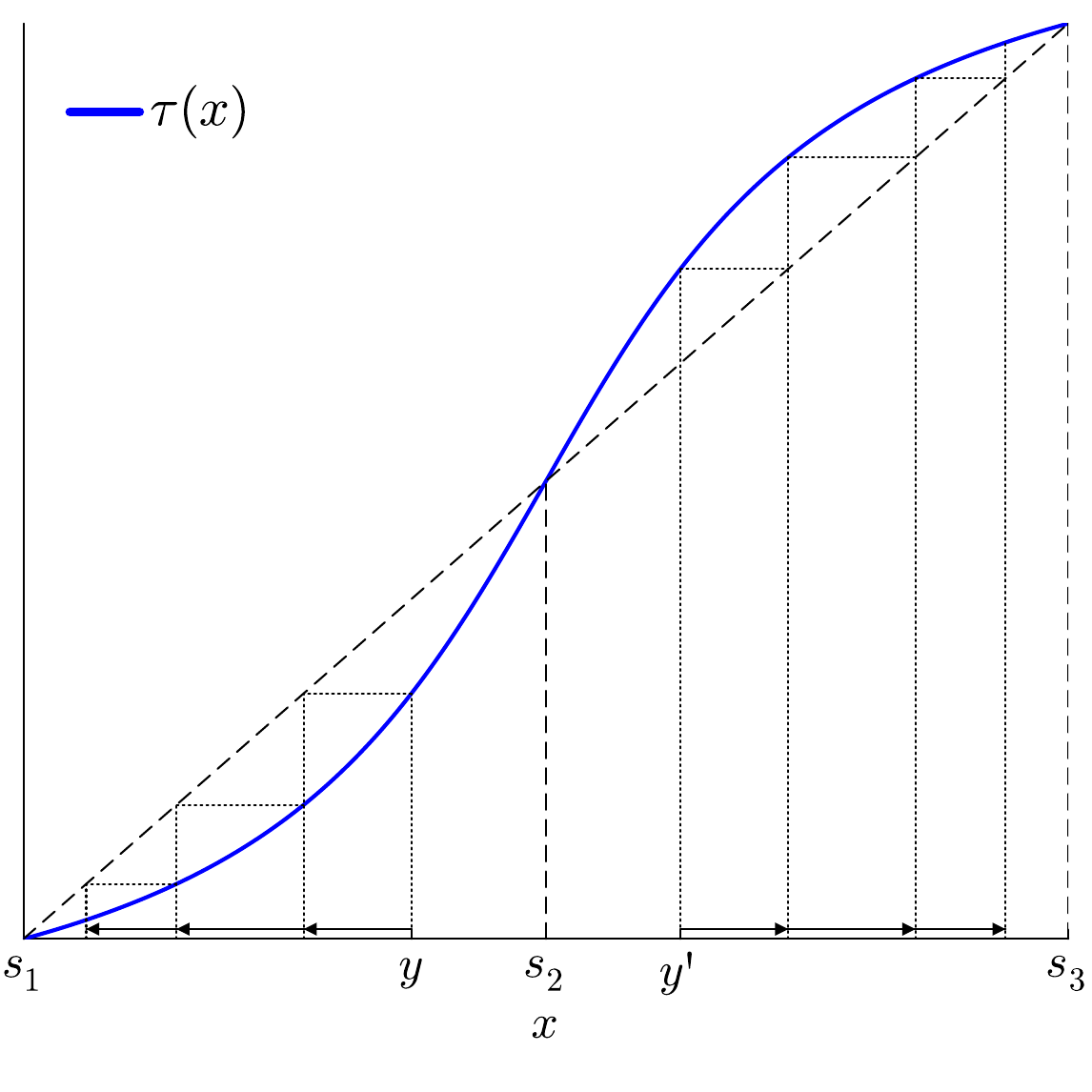}
    \caption{A transmission function $\tau : \mathcal{X} \to \mathcal{X}$ where $\mathcal{X} = [s_1,s_3]$ with three fixed points: $s_1$ and $s_3$ are sinks but $s_2$ is a source. Here, $\mathcal{X}_0$ and $\mathcal{X}_3$ are empty, and $\mathcal{X}_1 = [s_1,s_2]$ and $\mathcal{X}_2 = [s_2,s_3]$. An initial degenerate distribution at a capital level $y$ in the interior of $\mathcal{X}_1$ converges to a degenerate distribution at $s_1$. Similarly, an initial degenerate distribution at a capital level $y'$ in the interior of $\mathcal{X}_2$ converges to a degenerate distribution at $s_3$.}
    \label{fig:atoms}
\end{figure}

\subsection{Atomless initial distribution}\label{sec:densities}
In this section we analyze the conditions under which the population process possesses an atomless steady state distribution. We also characterize the dynamics when the initial distribution is atomless. Our results on the existence of atomless steady states and the conditions under which we obtain convergence to such states are entirely new and they are the central contribution of our paper.

Proposition \ref{prop:density_dynamic} below shows how densities evolve from one generation to the next. Let $\rho:\mathcal{X} \rightarrow \mathcal{X}$ denote the inverse of the transmission function, so $\rho(x) := \tau^{-1}(x)$ for each $x \in \mathcal{X}$. 

\begin{proposition}\label{prop:density_dynamic}
For any $t\geq 0$, if $F_{t}$ admits a density $f_{t}$ then $F_{t+1}$ admits a density $f_{t+1}$ which, for each $x \in \mathcal{X}$, satisfies
\begin{equation}\label{eq:evolution}
    f_{t+1}(x) = \frac{1}{\mathbb{E}_{t}[n]}  \frac{n(\rho(x))}{\tau'(\rho(x))} f_{t}(\rho(x)).
\end{equation}
\end{proposition}
\noindent This implies that if the initial distribution $F_{0}$ admits a density $f_{0}$ then the population process can be characterized by the evolution of a sequence of densities $(f_{t})_{t \geq 0}$. 
\begin{remark}\label{cor:growth}
If $F^*$ is a steady state of the population process that admits a density $f^*$ for which $f^*(s)>0$ and $s$ is a fixed point of $\tau$ then $\mathbb{E}^*[n] = n(s)/\tau'(s)$. This follows directly from evaluating \eqref{eq:evolution} in steady state at $s$.
\end{remark}
Next, we establish conditions guaranteeing the existence of atomless steady states and convergence to such steady states. To do so, we now introduce the notions of well-behavedness, niceness, and genericity, and we provide an in-depth discussion of these notions after the statement of Theorem \ref{thm} below. For any pair $x,y \in \mathcal X$ define 
\[ G_t(x, y) := \prod_{i=1}^t \frac{n(\rho^{[i]}(x))}{\tau'(\rho^{[i]}(x))} \frac{\tau'(y)}{n(y)} .\]
\noindent The expression $G_t$ plays a central role in our analysis. We say that the primitives $n$ and $\tau$ are \defn{well-behaved} if for any $y \in \mathcal{X}$, $G_t(\cdot,y)$ is integrable over $\mathcal{X}$ for some $t\geq 1$. For each $k \in \{0,\dots,K\}$ we say that the primitives are \defn{nice} in $\mathcal{X}_k$ if there is a fixed point of the transmission function $s \in \mathcal{X}_k$ such that \emph{either} $s$ is a sink, $s$ is not a stationary point of $n(\cdot)$, and $n(\cdot)$ is endpoint maximal at $s$ in $\mathcal{X}_k$, \emph{or} $s$ is a source, $s$ is not a stationary point of $n(\cdot)/\tau'(\cdot)$, and $n(\cdot)/\tau'(\cdot)$ is endpoint maximal at $s$ in $\mathcal{X}_k$.\footnote{Note that the `either' and `or' parts of the statement are mutually exclusive.} Finally, we say that the primitives $n$ and $\tau$ are \defn{generic} if the set
\begin{equation}\label{eq:generic}
\argmax_{k \in \{1,\dots,K\}} \; n(s_k)\mathbbm{1}[\text{$s_k$ is a sink}] + \frac{n(s_k)}{\tau'(s_k)}\mathbbm{1}[\text{$s_k$ is a source}]
\end{equation}
is a singleton, in which case we denote its unique element by $k^*$.

Theorem \ref{thm} is the central result of our paper.

\begin{theorem}\label{thm}
Suppose the primitives $n$ and $\tau$ are well-behaved, generic, and that they are nice in $\mathcal{X}_k$ for each $k \in \{0,\dots,K\}$.
\begin{enumerate}[leftmargin=0.7cm, label=(\roman*)]
    \item If $s_{k^*}$ is a source then there exists a steady state supported on $\mathcal{X}_{k^*-1} \cup \mathcal{X}_{k^*}$ with density
\[
f^*(x) \propto \lim_{t \to \infty} G_t(x,s_{k^*}) .
\]
If, moreover, the initial distribution $F_{0}$ admits a continuous and bounded density $f_{0}$ that is supported on $\mathcal X$, then for all $x \in \mathcal X$, $\lim_{t \to \infty} f_{t}(x) = f^*(x)$.
\item If, instead, $s_{k^*}$ is a sink and the initial distribution $F_{0}$ admits a continuous and bounded density $f_{0}$ that is supported on $\mathcal X$, then the population process converges to a distribution that is degenerate at $s_{k^*}$.
\end{enumerate}
\end{theorem}
Theorem \ref{thm} establishes sufficient conditions on the primitives $n$ and $\tau$ for the existence of an atomless steady state, and for the convergence of the population process to a distribution that is either atomless---Theorem \ref{thm} (i)---or degenerate---Theorem \ref{thm} (ii).

\begin{remark}\label{rem:flat}
The atomless steady states identified in Theorem \ref{thm} (i) are novel and do not arise when fertility is flat. With flat fertility of the form $n(x)=C$ for some $C>0$, the population process \eqref{main_nonlin} reduces to $F_{t+1}(x) = F_t(\rho(x))$ which, through iteration, gives us
\[
F_{t}(x) = F_0(\rho^{[t]}(x)) .
\]
In turn, this implies that for any continuous initial distribution $F_0$, any $k \in \{0,\dots,K\}$, and any source $s \in \mathcal{X}_k$, $\lim_{t \to \infty} F_t(x) = F_0(s)$ for each $x \in \interior \mathcal{X}_k$. So, with flat fertility and any continuous initial distribution, the cumulative distribution flattens out around sources and therefore accumulates into atoms at sinks of the transmission function.
\end{remark}
We give an in-depth discussion of Theorem \ref{thm} in Sections \ref{sec:discussionAs} and \ref{sec:discussionThm}. Some intuition for the roles of genericity, well-behavedness, and niceness is given in Section \ref{sec:discussionAs}, and we argue that these conditions are easy to check. In Section \ref{sec:discussionThm}, we provide intuition for the dynamics of our system and, in particular, for the emergence of the atomless steady states under the conditions of Theorem \ref{thm} (i), by considering several specific parametric examples. Section \ref{sec:discussionThm} is important for two reasons: first, each example highlights a different aspect of Theorem \ref{thm} and helps to understand the dynamics of the population process and the roles of genericity, well-behavedness, and niceness in determining these dynamics. Second, the examples show that the conditions of genericity, well-behavedness and niceness are satisfied in `natural' settings (in which the fertility and transmission functions take shapes that one might ordinarily encounter in an economic model). In fact, Section \ref{sec:endo} goes a step further in showing that the conditions are satisfied even in a micro-founded model with endogenous fertility and transmission.  

The proof of Theorem \ref{thm}, all steps of which are in the appendix, consists in decomposing the population process into $K+1$ parallel processes, each of which is analyzed separately, and then reconstructing the overall distribution of capital from the parallel processes. More specifically, we rely on the insight that capital in a dynasty cannot `cross' a fixed point of the transmission function: if a person has a capital level in $\mathcal{X}_k$ for some $k$ then the capital of every descendant of that person will also be in $\mathcal{X}_k$. This implies that the evolution of the shape of the distribution of capital truncated to an interval $\mathcal{X}_k$ is independent of the evolution of its shape when truncated to any other interval $\mathcal{X}_{k'}$. We therefore truncate the distribution of capital to each of the $K+1$ intervals and track the evolution of each truncated distribution separately, from which we reconstruct the untruncated distribution over all of $\mathcal{X}$.

\subsubsection{Genericity, well-behavedness, and niceness}\label{sec:discussionAs}
Theorem \ref{thm} is predicated on genericity, well-behavedness, and niceness, which are restrictions only on $n$ and $\tau$ that are easy to verify in practice, as we argue below.

\paragraph{Genericity} This condition is straightforward to check and it is innocuous for two reasons. First, we impose it in Theorem \ref{thm} only for simplicity: the statement of the result would be unwieldy without it, and the results in the appendix provide the tools to deal with more general cases (see Proposition \ref{prop:parallel}). Second, the condition holds generically. Indeed, it holds trivially whenever $\tau$ has a unique fixed point and, even if the transmission function has multiple fixed points, we would generically expect the argmax in \eqref{eq:generic} to be a singleton in the presence of differential fertility.

Recall, as mentioned above, that the evolution of the shape of the distribution of capital truncated to an interval $\mathcal{X}_k$ is independent of the evolution of its shape truncated to any other interval $\mathcal{X}_{k'}$. However, the population growth rates across different intervals matter for the shape of the untruncated distribution of capital: if the population growth rate in the process truncated over some interval $\mathcal{X}_k$ is much smaller than over some other interval $\mathcal{X}_{k'}$, then the mass of people in the interval $\mathcal{X}_k$ must eventually go to zero since, in the overall distribution, it will dominated by the faster growing population in $\mathcal{X}_{k'}$. Genericity is a condition on the long-run population growth rates of the $K+1$ parallel process ensuring that the mass of people goes to zero everywhere except for capital levels in $\mathcal{X}_{k^*-1} \cup \mathcal{X}_{k^*}$ when $s_{k^*}$ is a source and that the distribution accumulates into an atom exactly and only at $s_{k^*}$ when $s_{k^*}$ is a sink. Without genericity, it would be possible for the long-run density to be supported over multiple non-adjacent intervals $\mathcal{X}_k$, or for the distribution of capital to accumulate into several atoms (at more than one fixed point of the transmission function).

\paragraph{Well-behavedness} This is a technical condition that requires the integrability of $G_t(\cdot,y)$ for some $t \geq 1$. In the appendix, we show that this condition is, jointly with the other assumptions of the theorem, sufficient for the limit $\lim_{t\to\infty} G_t(\cdot,y)$ not only to exist but also to be integrable. Integrability of this function limit is a necessary condition for the steady state $f^*$ in Theorem \ref{thm} (i) to in fact be a density, and it also implies that the corresponding cumulative distribution function $F^*$ is continuous. The integrability condition is similarly needed for Theorem \ref{thm} (ii).

Well-behavedness is easy to check thanks to Propositions \ref{prop:int1} and \ref{prop:int2} below.
\begin{proposition}\label{prop:int1}
If $n$ is integrable on $\mathcal{X}$ then the primitives are well-behaved.
\end{proposition}
\noindent Integrability of the fertility function is therefore an easy-to-verify sufficient condition for well-behavedness. However, as the result below shows, $n$ being integrable is not a necessary condition.
\begin{proposition}\label{prop:int2}
Let $\mathcal{X}=[1,\infty)$. If the transmission function $\tau: \mathcal{X} \rightarrow \mathcal{X}$ is given by $\tau(x) = x^a$ for some positive $a \neq 1$, and the fertility function $n :\mathcal{X} \rightarrow (0,\infty)$ satisfies the condition that there is some $m>0$ such that $\lim_{x \rightarrow \infty}x^m n(x) = M$ where $0\leq M<\infty$, then the primitives are well-behaved.
\end{proposition}
\noindent Proposition \ref{prop:int2} shows that even if $n$ is not integrable over $\mathcal{X}$, it is still possible for the primitives to be well-behaved. For example, $n(x) \propto 1/\sqrt{x}$ is not integrable over $[1,\infty)$ but it satisfies the condition of Proposition \ref{prop:int2}.

\paragraph{Niceness} This is a condition on the functions $n(\cdot)$ and $n(\cdot)/\tau'(\cdot)$ which is straightforward to verify.\footnote{The `either' and `or' parts of our niceness condition, while mutually exclusive, are not exhaustive. In other words, there can be situations in which the primitives satisfy neither the `either' nor the `or' part of the condition. Characterizing outcomes in such cases remains an open question.} As discussed in the appendix,\footnote{For example, c.f. footnote \ref{fn:tech}.} the requirement that the fixed point $s$ is not a stationary point of these functions (i.e. that the functions have a non-zero derivative at $s$), which holds generically, is a technical condition that avoids pathological cases and ensures that $\lim_{t\to\infty} G_t(\cdot,s)$ exists. The endpoint maximality condition is a condition on the steady state population growth factor. Note that the condition is about endpoints \emph{only}. For example, when we require that $n(\cdot)/\tau'(\cdot)$ is endpoint maximal at $s$ in $\mathcal{X}_k$, we make no assumptions whatsoever about the values of $n(x)/\tau'(x)$ for $x \in \interior \mathcal{X}_k$. At $x \in \interior \mathcal{X}_k$, $n(x)/\tau'(x)$ can be non-monotonic and take values that are smaller or larger than at $s$.

\subsubsection{Parametric examples}\label{sec:discussionThm}
Next, we look at five parametric examples, Examples A-E, below. In each example, Assumptions \ref{as1} and \ref{as2} both hold. Each example highlights a different aspect of Theorem \ref{thm} and helps to understand the dynamics of the population process. The space of possible capital levels $\mathcal{X}$ varies from one example to another but, in every example, we set the initial density of capital $f_0$ to be  continuous, bounded, and supported on $\mathcal{X}$, as required by Theorem \ref{thm}.\footnote{In the figures corresponding to each of our examples, we set $f_0$ to be a mixture of two Gaussian distributions truncated to live in the space $\mathcal{X}$. The specific choice of the parameters of the Gaussians is immaterial for our purposes since Theorem \ref{thm} applies to any initial density that is continuous, bounded, and supported on $\mathcal{X}$.}

\begin{remark}\label{rem:empirical_validity}
    We do not claim that the functional forms for $n$ and $\tau$ that we choose in the examples below are empirically valid. In fact, whatever is an appropriate functional form must depend on the precise interpretation that is given to the `capital' variable that is being transmitted across generation. Indeed, the empirically appropriate shapes of $n$ and $\tau$ will depend on whether capital is taken to be income, or wealth, or some other variable, and we have so far not committed to any particular interpretation of capital. The purpose of the examples below is to provide intuition and to showcase applications of Theorem \ref{thm} in simple settings.
\end{remark}

\paragraph{Example A.} This first example provides some intuition for the dynamics of the population process and highlights the new type of atomless steady states that are identified in Theorem \ref{thm} part (i). Suppose that $\mathcal{X} = [0,\infty)$, that fertility $n(x)$ is proportional to $e^{-mx}$, and that the transmission function is given by $\tau(x) = ax$. We examine three cases, all of which are illustrated in Figure \ref{fig:exponential}.

\begin{figure*}[t!]
    \centering
    \caption*{Example A}
    \begin{subfigure}[t]{0.325\linewidth}
        \centering
        \caption*{(a)}
        \includegraphics[width=1\linewidth]{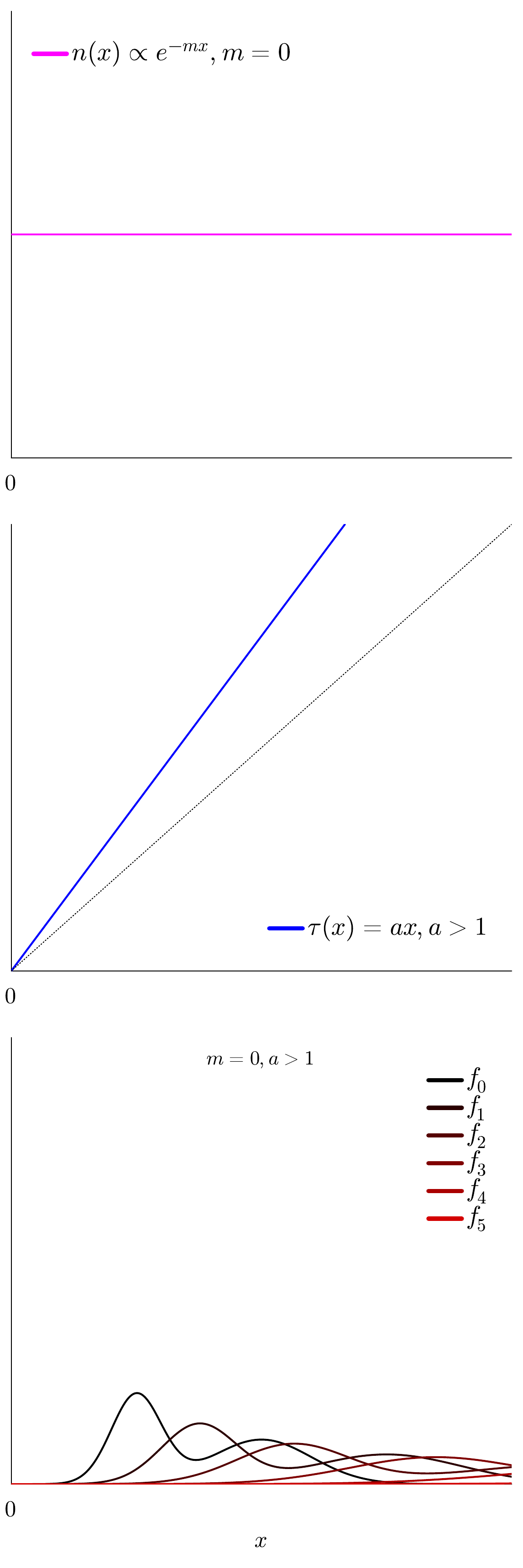}
    \end{subfigure}%
    \hspace{0.03cm}
    \begin{subfigure}[t]{0.325\linewidth}
        \centering
        \caption*{(b)}
        \includegraphics[width=1\linewidth]{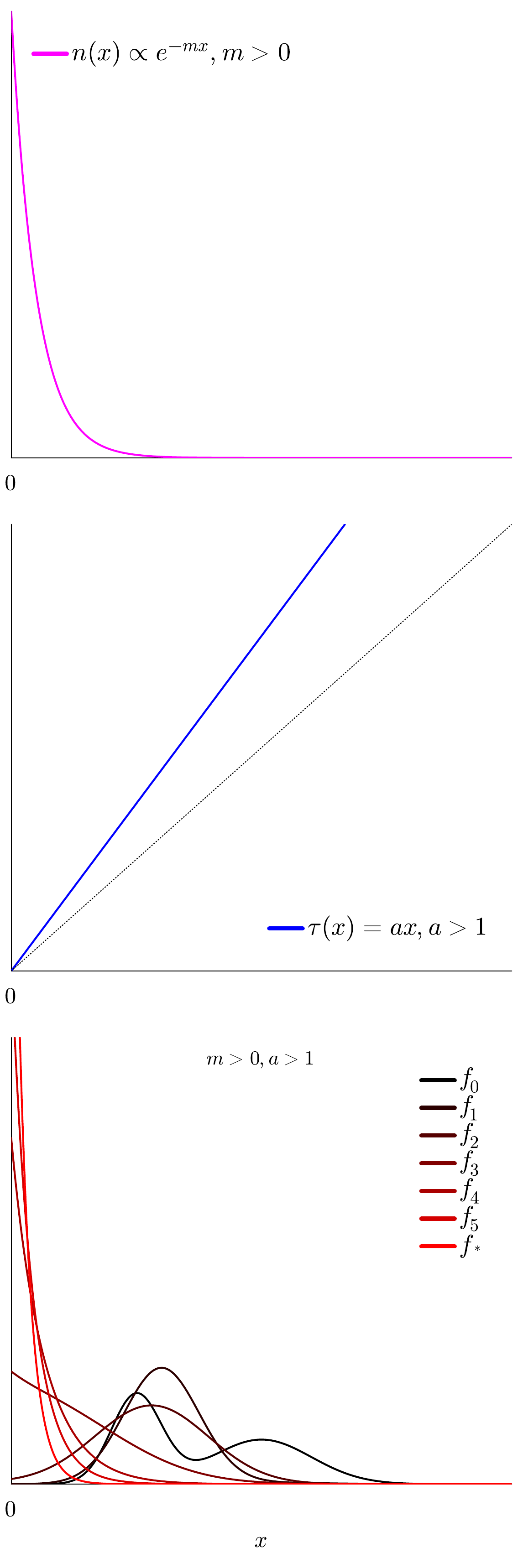}
    \end{subfigure}
    \begin{subfigure}[t]{0.325\linewidth}
        \centering
        \caption*{(c)}
        \includegraphics[width=1\linewidth]{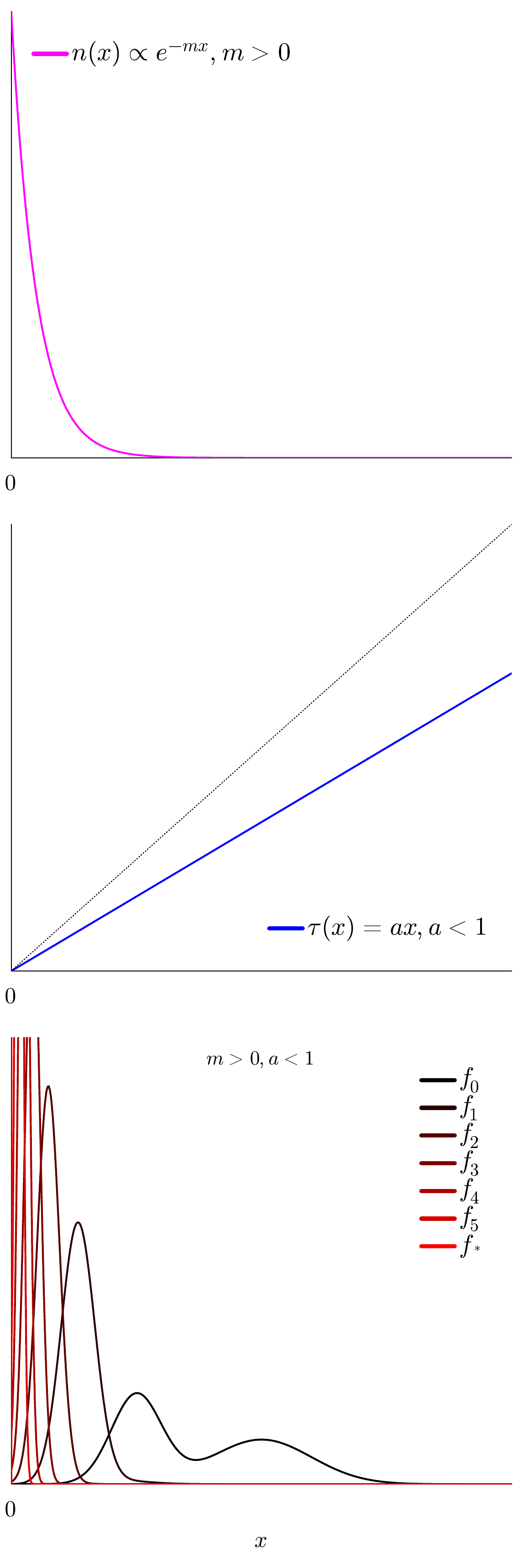}
    \end{subfigure}
\caption{Fertility, transmission, and the evolution of the distribution of capital for Example A.}
\label{fig:exponential}
\end{figure*}

In cases (a) and (b), we set $a > 1$ so that the transmission function has a unique fixed point at $s_1=0$ and it is a source. This implies that each child's capital is greater than their parent's capital. As expected, when fertility is flat, this leads to an explosive process in which capital is pushed towards ever greater levels: as can be seen in case (a), where we set $m=0$, the density of capital escapes to the right. In contrast, consider case (b) in which we set $m>0$, so that fertility is decreasing with capital. Now there are two opposing forces. On the one hand, capital is pushed towards higher levels due to the transmission function but, on the other hand, people at higher levels of capital reproduce less rapidly. This process converges to an atomless steady state distribution of capital whose density $f^*$ is shown in Figure \ref{fig:exponential}. This shows that it would be incorrect, in general, to presume that an intergenerational transmission function that would lead to explosive dynamics under flat fertility must also result in explosive dynamics under differential fertility.

In case (c), we set $a < 1$ so that the transmission function has a unique fixed point at $s_1=0$ and it is a sink. Each child's capital is now less than their parent's capital. Regardless of whether fertility is flat or decreasing, this results in a distribution of capital that eventually accumulates into a single atom at $x=0$.

We now employ Theorem \ref{thm} to both explain the dynamics in cases (b) and (c), and to provide an explicit expression for the atomless steady state distribution of case (b). First, if $m>0$ the fertility function, which is proportional to $e^{-mx}$ is integrable over $\mathcal{X}$. By Proposition \ref{prop:int1}, the primitives are well-behaved. It is also easy to see that the primitives are nice and, since the fixed point of $\tau$ is unique when $a \neq 1$, they are also generic. If $a<1$, as in case (c), the unique fixed point of $\tau$ is a sink so we know by Theorem \ref{thm} (ii) that the distribution of capital converges to a degenerate steady state at $s_1=0$. If $a>1$, as in case (b), the unique fixed point of $\tau$ is a source so we know by Theorem \ref{thm} (i) that the distribution of capital converges to an atomless steady state distribution supported on $\mathcal{X}$ whose density is given by
\[
f^*(x) \propto \lim_{t \to \infty} G_t(x,s_1)= \lim_{t \to \infty} \prod_{i=1}^t  \frac{\exp \left\{ -m \rho^{[i]}(x) \right\}}{\exp \left\{ - m s_1 \right\}} \propto \exp \left\{ -\frac{m}{a-1}  x \right\}
\]
where $\rho^{[i]}(x) = x a^{-i}$. In other words, the steady state distribution is exponential with parameter $m/(a-1)$. Observe that Theorem \ref{thm} establishes convergence to this steady state distribution from essentially any initial distribution.

\paragraph{Example B.} Unlike Example A, this example considers a case in which the transmission function has multiple fixed points. Here, we let $\mathcal{X} = [-5,5]$ and the transmission function is given by $\tau(x) = cx + \text{tanh}(x)$ where $c = 1 - \text{tanh(5)/5} \approx 0.8$, thereby resulting in three fixed points, at $s_1 = -5$, $s_2 = 0$, and $s_3 = 5$, with the middle one being a source and the other two being sinks. This transmission function is illustrated in each panel of Figure \ref{fig:B}.\footnote{Note that $\tau'(x) = c + \text{sech}(x)^2$. The hyperbolic tangent function is defined as $\text{tanh}(x):= (e^x - e^{-x})/(e^x + e^{-x})$ and the hyperbolic secant function is defined as $\text{sech}(x):= 2/(e^x + e^{-x})$. The inverse transmission, $\rho$, does not have an analytical expression, but we computed it numerically to generate the densities shown in Figure \ref{fig:B}.} 

In panels (a) and (b) fertility takes a shape that is proportional to a skewed Gaussian distribution, and fertility is flat in panel (c). In all cases, fertility is a bounded function which is therefore integrable over $\mathcal{X}$ since the latter is compact. By Proposition \ref{prop:int1} our primitives are therefore well-behaved. Observe that in panels (a) and (b), fertility is `hump-shaped'.\footnote{See for example \cite*{jones08} and \cite*{vogl16b} for empirical regularities in the shape of fertility functions.}

Let's start with case (a). Here, $n(s_2)/\tau'(s_2) > n(s_k)/\tau'(s_k) > n(s_k)$ for $k \in \{1,3\}$ and $s_2$ is not a stationary point of $n(\cdot)/\tau'(\cdot)$. The primitives are therefore nice and generic over $\mathcal{X}_k$ for each $k$, and the unique argmax of \eqref{eq:generic} is at $k^*=2$. Theorem \ref{thm} part (i) allows us to conclude that there is an atomless distribution with density $f^*$ supported on $\mathcal{X}$ to which the system converges. There is no explicit analytic expression for this density but our result guarantees the existence of such a limiting density to which the population process converges from essentially any initial distribution. Observe that Figure \ref{fig:B} shows that the density after 10 and after 50 generations is essentially unchanged. The dynamics here work as follows: the transmission drives capital away from $s_2=0$ towards $s_1=-5$ and $s_3=5$, but this effect is counteracted by fertility being greater at $s_2=0$ than at the endpoints of $\mathcal{X}$, and this is what results in a stable atomless long-run distribution of capital.

In case (b), $n(s_1) > n(s_2) > n(s_2)/\tau'(s_2) > n(s_3)/\tau'(s_3) > n(s_3) $, and none of the fixed points of $\tau$ is a stationary point of $n(\cdot)$ or $n(\cdot)/\tau'(\cdot)$. The primitives are therefore nice and generic over $\mathcal{X}_k$ for each $k$, and the unique argmax of \eqref{eq:generic} is at $k^*=1$. Theorem \ref{thm} part (ii) allows us to conclude that the limiting distribution is degenerate at $s_1$. What drives the dynamics here is this: the transmission function drives capital away from $s_2=0$ towards $s_1=-5$ and $s_3=5$. The distribution of capital would therefore accumulate at $s_1=-5$ and $s_3=5$ but fertility is lower at $s_3$ than at $s_1$, which is why the distribution of capital eventually accumulates into an atom only at $s_1$.

Case (c) illustrates Remark \ref{rem:flat}. With flat fertility, the distribution of capital accumulates into atoms at the sinks $s_1=-5$ and $s_3 = 5$. Since the population growth factor at each of these fixed points is the same, we end up with a long-run distribution with two point masses, one at each of these fixed points. This recalls the dynamics of some poverty trap models with flat fertility.

\begin{figure*}[t!]
    \centering
    \caption*{Example B}
    \begin{subfigure}[t]{0.325\linewidth}
        \centering
        \caption*{(a)}
        \includegraphics[width=1\linewidth]{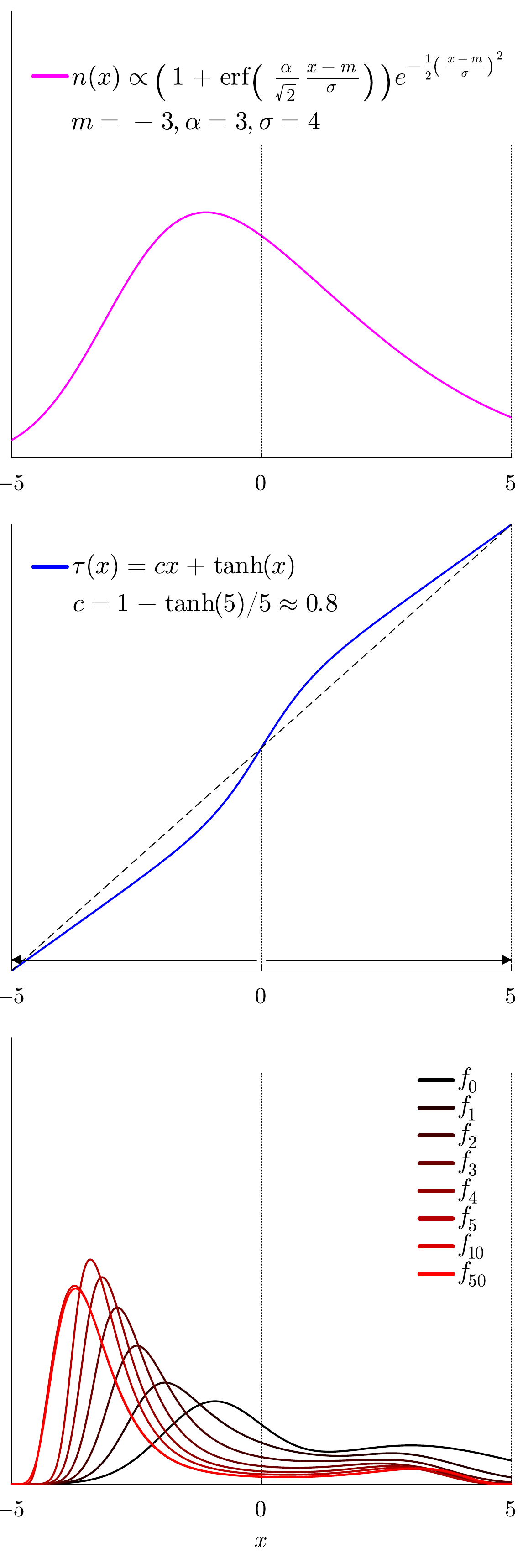}
    \end{subfigure}%
    \hspace{0.03cm}
    \begin{subfigure}[t]{0.325\linewidth}
        \centering
        \caption*{(b)}
        \includegraphics[width=1\linewidth]{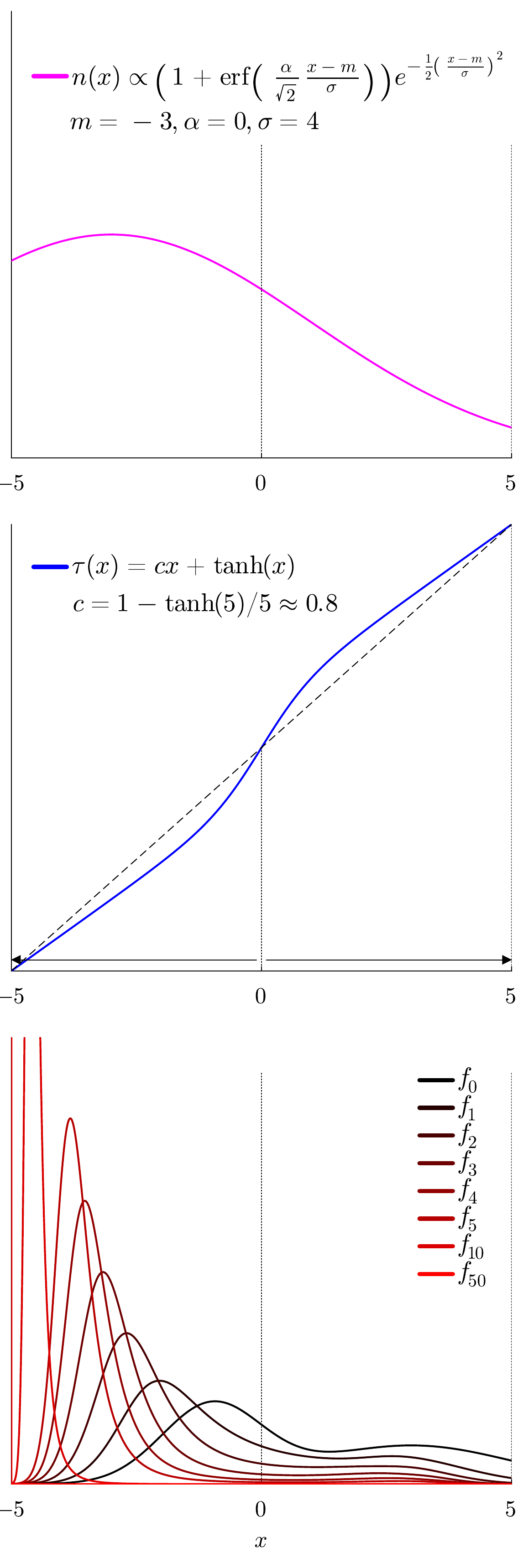}
    \end{subfigure}
    \begin{subfigure}[t]{0.325\linewidth}
        \centering
        \caption*{(c)}
        \includegraphics[width=1\linewidth]{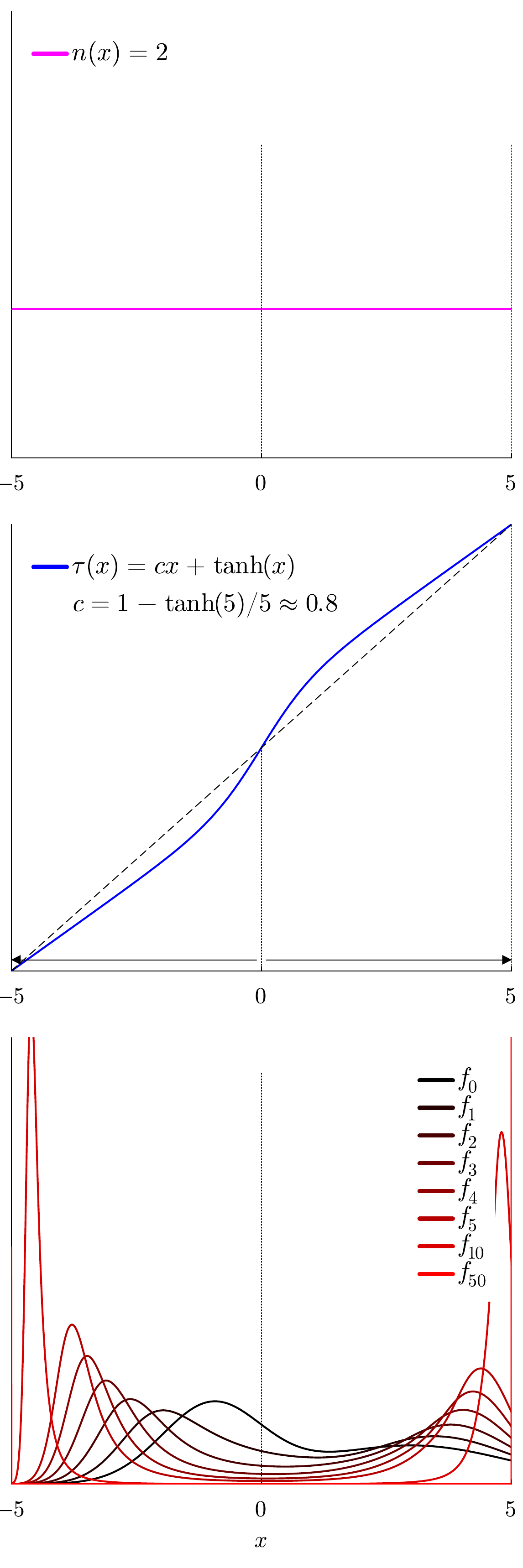}
    \end{subfigure}
\caption{Fertility, transmission, and the evolution of the distribution of capital for Example B.}
\label{fig:B}
\end{figure*}

\begin{figure*}[t!]
    \centering
    \begin{subfigure}[t]{0.32\linewidth}
        \centering
        \caption*{Example C}
        \includegraphics[width=1\linewidth]{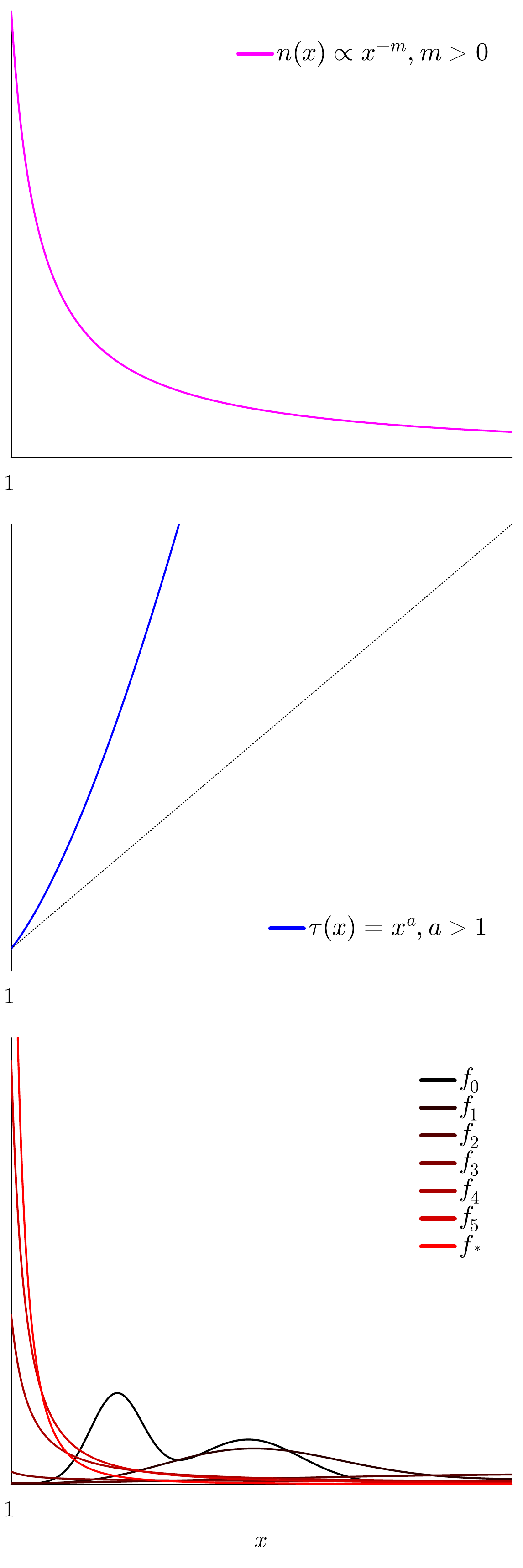}
    \end{subfigure}%
    \hspace{0.03cm}
    \begin{subfigure}[t]{0.32\linewidth}
        \centering
        \caption*{Example D}
        \includegraphics[width=1\linewidth]{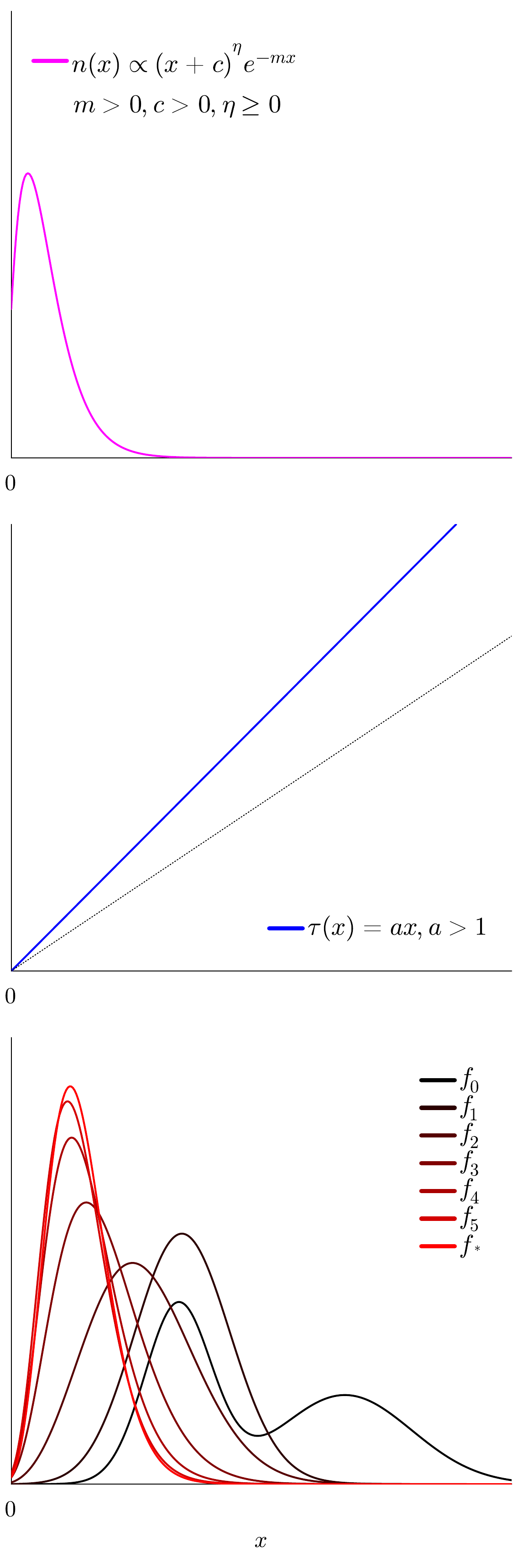}
    \end{subfigure}
    \begin{subfigure}[t]{0.32\linewidth}
        \centering
        \caption*{Example E}
        \includegraphics[width=1\linewidth]{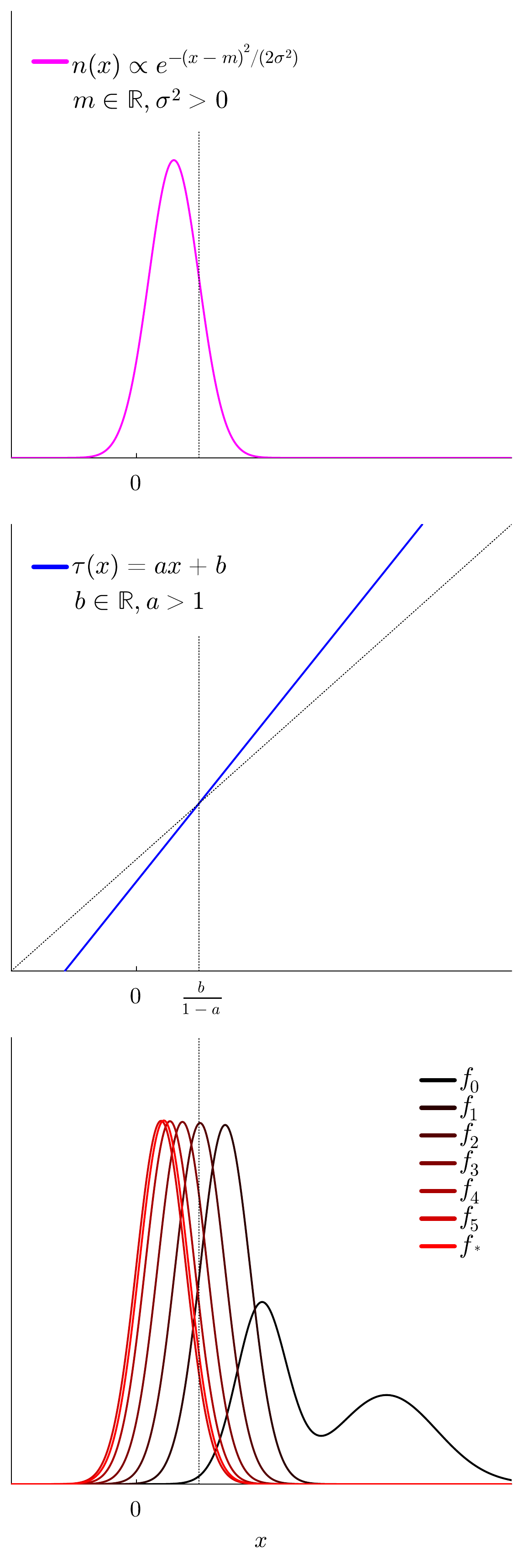}
    \end{subfigure}
\caption{Fertility, transmission, and the evolution of the distribution of capital for Examples C (on the left), D (in the middle), and E (on the right).}
\label{fig:others}
\end{figure*}

\paragraph{Example C.} We now consider a situation in which the fertility function is not necessarily integrable. Suppose that $\mathcal{X} = [1,\infty)$, that fertility $n(x)$ is proportional to $x^{-m}$ with $m>0$, and that the transmission function is given by $\tau(x) = x^a$ with $a>1$. These are illustrated in the left side panel Figure \ref{fig:others}. Let $s_1=1$ denote the unique fixed point of the transmission function and observe that it is a source. Even though the fertility function may not be integrable over $\mathcal{X}$ for each $m>0$, we know by Proposition \ref{prop:int2} that the primitives are well-behaved. It is also trivial to check that they are nice and generic. By Theorem \ref{thm} part (i), we conclude that the distribution of capital converges to an atomless steady state distribution supported on $\mathcal{X}$ whose density is given by
\[
f^*(x) \propto \lim_{t\to\infty} \prod_{i=1}^t  \left(\frac{\rho^{[i]}(x) }{ s_1 } \right)^{-(m+a-1)} \propto x^{-\left(\frac{m}{a-1} + 1\right)} 
\]
where $\rho^{[i]}(x) = x^{a^{-i}}$. Using Pareto[$c,\nu$] to denote a Pareto distribution with support $[c,\infty)$ and cumulative distribution function $1-\left( \frac{c}{x}\right)^\nu$ where $\nu>0$, we conclude that the steady state distribution is Pareto[$1,1 + m/(a-1)$].

\begin{remark}
    Given the importance of the Pareto distribution in the literature on economic inequality, we make several observations about Example C and the Pareto steady state that our deterministic system converges to in that example. Interest in the Pareto distribution arises from the fact that Pareto \emph{tails} are an empirical regularity in the distribution of wealth and of income \citep*{benhabib16,piketty15}. \cite*{champernowne1953model,wold1957model,benhabib11,gabaix15,beare2022determination} and others theoretically derive Pareto tails from an underlying (usually multiplicative) random process. In an entirely different stochastic framework with a pure-birth process and a constant branching rate, \cite{yule1925ii} also derives Pareto tails. Example C shows that we can also generate Pareto tails,\footnote{In fact, Example C generates a distribution of capital that is fully Pareto rather than just having a power-law tail. As observed in Remark \ref{rem:empirical_validity}, whether this is empirically valid depends on the interpretation of the capital variable. That said, one could modify the shapes of $n$ and $\tau$ at low values of capital to obtain a non-Pareto shape at the bottom of the distribution of capital while retaining a power-law shape at the top.} but what we have here is distinct since the Pareto tail arises from a purely deterministic process. Moreover, differential fertility is a key driver of long-run inequality in our case (for example, the model of \citealp{wold1957model} features flat fertility, and it is Poisson deaths that drive the Pareto tails in their case). Finally, unlike the aforementioned papers for which a key aim is to explain a specific empirical regularity, our model can generate distributions belonging to many different (possibly non-Pareto) parametric families.
\end{remark}

\paragraph{Example D.} Suppose that $\mathcal{X} = [0,\infty)$, that fertility $n(x)$ is proportional to $(x+c)^\eta e^{-mx}$ with $m>0$, $c>0$, and $\eta \geq 0$, and that the transmission function is given by $\tau(x) = ax$ with $a>1$. These are illustrated in the middle panel of Figure \ref{fig:others}. Let $s_1=0$ denote the unique fixed point of the transmission function and observe that it is a source. Thanks to its exponential part, the fertility function is integrable on $\mathcal{X}$. The primitives are therefore well-behaved thanks to Proposition \ref{prop:int1}, and it is easy to check that they are also nice and generic.\footnote{Note that $n(\cdot)/\tau'(\cdot)$ is non-monotonic and  does not reach its maximum at $s_1$, thereby illustrating the fact that only endpoint maximality matters for niceness.} By Theorem \ref{thm} part (i) we conclude that the distribution of capital converges to an atomless steady state distribution supported on $\mathcal{X}$ whose density is given by
\[
f^*(x) \propto  \lim_{t\to\infty} \prod_{i=1}^t  \exp \left\{ - m \rho^{[i]}(x) \right\} \left(\frac{\rho^{[i]}(x) + c }{c } \right)^{\eta} \propto  \exp \left\{ - \frac{m}{a-1} x \right\}   \left( \frac{(-x/c ; 1/a)_\infty}{1+x/c}\right)^{\eta} 
\]
where $\rho^{[i]}(x) = x a^{-i}$ and $(x;q)_t := \prod_{i=0}^{t-1} (1-xq^i)$ is the $q$-Pochammer symbol.\footnote{To see why $(-x/c ; 1/a)_\infty$ converges, notice that $\prod_{i=0}^\infty \left(1 + \frac{x}{c} a^{-i} \right)$ converges if and only if $\frac{x}{c}\sum_{i=0} ^\infty a^{-i}$ does, but this latter series is equal to $\frac{x}{c (a-1)}$ whenever $a >1$.} We are not aware of a standard nomenclature for this particular distribution.

\paragraph{Example E.} Suppose that $\mathcal{X} = (-\infty,\infty)$, that fertility $n(x)$ is proportional to $e^{-(x-m)^2/(2\sigma^2)}$ with $m \in \mathbb{R}$ and $\sigma^2>0$, and that the transmission function is given by $\tau(x) = ax+b$ with $b\in \mathbb{R}$ and $a>1$. Let $s_1=b/(1-a)$ denote the unique fixed point of the transmission function and observe that it is a source. Like previous examples, it is easy to check that the primitives are well-behaved, nice and generic. By Theorem \ref{thm} part (i) we conclude that the distribution of capital converges to an atomless steady state distribution supported on $\mathcal{X}$ whose density is given by
\[
f^*(x) \propto \lim_{t\to\infty}  \prod_{i=1}^t \frac{\exp \left\{ - \frac{(\rho^{[i]}(x) - m)^2}{2 \sigma^2} \right\}}{\exp \left\{ - \frac{(s_1 - m)^2}{2 \sigma^2} \right\}} \propto \exp \left\{ - \frac{\left( x - m^*\right)^2}{2 \sigma^2(a^2 - 1)} \right\} 
\]
where $\rho^{[i]}(x) = xa^{-i} +s (1 - a^{-i})$ and $m^* = (1+a)m - as$. In other words, the steady state distribution is Gaussian with mean $m^*$ and variance $\sigma^2 (a^2-1)$.\footnote{With a straightforward transformation of variables, Example E can be modified to obtain a long-run distribution of capital that is lognormal rather than Gaussian.}

\subsection{Comparative statics}\label{sec:comp_stats}
We now derive general comparative statics for the distribution of capital.

A function $\hat{\mu}: \mathcal{X} \rightarrow (0,\infty)$ \defn{MLR-dominates} $\mu: \mathcal{X} \rightarrow (0,\infty)$ (i.e. $\hat{\mu}$ dominates $\mu$ in the monotone likelihood ratio order), denoted as $\hat{\mu} \succsim \mu$, if the ratio $\hat{\mu}(x)/\mu(x)$ is non-decreasing for every $x \in \mathcal X$. And $\hat{\mu}$ \defn{strictly MLR-dominates} $\mu$, denoted $\hat{\mu} \succ \mu$, if the ratio is, additionally, increasing for some $x \in \mathcal X$.

Consider two population processes $( (F_{t})_{t \geq 0},n,\tau )$ and $( (\hat{F}_{t})_{t \geq 0},\hat{n},\tau )$ with identical transmission functions but possibly different fertility functions $n$ and $\hat{n}$, and initial distributions $F_0$ and $\hat{F}_0$.

\begin{proposition}\label{prop:comp1}
Suppose $\hat{F}_{0}$ and $F_{0}$ admit densities $\hat{f}_{0}$ and $f_{0}$ respectively. Suppose that $\hat{n} = n$ but $\hat{f}_{0} \succ f_{0}$. Then for all $t \geq 0$, $\hat{F}_{t}$ first-order stochastically dominates $F_{t}$.

\noindent If, additionally, all primitives are well-behaved, nice in $\mathcal{X}_k$ for each $k\in\{0,\dots,K\}$, and generic with $s_{k^*}$ being a source, then there exist atomless steady state distributions $\hat{F}^*$ and $F^*$ to which the population processes respectively converge, and $\hat{F}^*=F^*$.
\end{proposition}
Proposition \ref{prop:comp1} shows that if the two processes have the same fertility function but the initial density of one process MLR-dominates the initial density of the other, then the distributions of capital are ordered by first-order stochastic dominance in every generation.\footnote{Recall that a distribution $\hat{F}$ first-order stochastically dominates $F$ if $\hat{F}(x) \leq F(x)$ for each $x \in \mathbb{R}$ and the inequality is strict for some $x$.} However, because the shape of the initial distribution does not affect the long-run outcomes, the two processes converge to the same steady state in the limit as $t \rightarrow \infty$. Such a relationship concerning the initial value dependence of the \emph{transitional} dynamics features in several studies of dynamic models with differential fertility (e.g. see \citealp*{delacroix03}). 

\begin{proposition}\label{prop:comp2}
Suppose $\hat{F}_{0}$ and $F_{0}$ admit densities $\hat{f}_{0}$ and $f_{0}$ respectively. Suppose that $\hat{f}_{0} = f_{0}$ but  $\hat{n} \succ n$. Then for all $t \geq 0$, $\hat{F}_{t}$ first-order stochastically dominates $F_{t}$

\noindent If, additionally, all primitives are well-behaved, nice in $\mathcal{X}_k$ for each $k\in\{0,\dots,K\}$, and generic with $s_{k^*}$ being a source, then there exist atomless steady state distributions $\hat{F}^*$ and $F^*$ to which the population processes respectively converge, and $\hat{F}^*$ first-order stochastically dominates $F^*$.
\end{proposition}
Proposition \ref{prop:comp2} shows that if the two processes have the same initial density but the fertility function of one process MLR-dominates the other's, then the distributions of capital are ordered by first-order stochastic dominance in every generation. Moreover, since the limiting distributions depend on the primitives (and in particular on the different fertility functions), the two processes converge to different steady state distributions that are also ordered by first-order stochastic dominance.

Proposition \ref{prop:comp2} is a novel contribution to the literature on the intergenerational transmission of inequality. The closest existing result is due to \cite*{chu90}. In a discretized version of our model they show that, under certain conditions, a decrease in the fertility of the lowest income group leads to a first order stochastic dominance shift in the steady state distribution if fertility was downward sloping to begin with and the decrease is sufficiently small. In our case, fertility is not restricted to be downward sloping and the fertility functions $\hat{n}$ and $n$ can differ at any level of capital. Our only requirement is that the functions are MLR-ordered, thereby allowing for changes \emph{anywhere} along the fertility curve.

There is a neat economic interpretation of Proposition \ref{prop:comp2} in terms of elasticity. Suppose that $\mathcal{X} \subseteq (0,\infty)$. Then $\hat{n} \succsim n$ if and only if for each $x \in \mathcal X$,\footnote{The restriction $\mathcal{X} \subseteq (0,\infty)$ is just for convenience. We would simply need to be more careful with the direction of the inequality if we allowed capital to take negative values.}
\[
x\frac{\hat{n}'(x)}{\hat{n}(x)} \geq x\frac{n'(x)}{n(x)}.
\]
In other words first-order stochastic dominance shifts in Proposition \ref{prop:comp2} can be interpreted as being due to changes in the \emph{capital elasticity of fertility}.

\section{Endogenous fertility and transmission} \label{sec:endo}

Our entire analysis has so far considered exogenously given and time-invariant fertility and transmission functions. In this section, we show how our results can be carried over to a model with endogenous fertility and transmission functions. To do so, we consider a simple growth model based primarily on \cite*{delacroix03,Croix2004Public} that features endogenous differential fertility and transmission functions in which parents face a Becker-style quantity-quality trade-off.

\subsection{Model}
Our economy consists of overlapping generations of people who live for two periods: childhood and adulthood. During childhood people get the level of education chosen by their parents and acquire a level of (human) capital determined by the amount of education they receive and by the capital level of their parent. Children make no decisions. During adulthood people earn a labour income that depends on their level of human capital, and make optimal decisions regarding the number of children to have, investments in education for the next generation, and own consumption.

Capital takes values in $(0,\infty)$. In each generation $t$, we assume that $z$ units of capital translate into $z$ `effective' hours of labour and the wage per effective hour of labour is given by $w_t$. Each parent with capital $z$ in generation $t$ inelastically supplies $z$ effective hours of labour and therefore has a labour income of $w_t z$. This simple relationship between (human) capital and labour income appears in many models of fertility choice \citep*{delacroix03,jones08,vogl16b}.

Parents care about their own consumption of a num\'{e}raire good and about the aggregate future income of their offspring. Preferences are quasi-linear and given by
\begin{equation}\label{eq:utility}
c + \ln(  n w_{t+1} z' ) ,
\end{equation}
where $n$ is the number of children, $c$ is consumption, $z'$ is the capital level of each child.\footnote{We could have a more general specification with a utility function of the form $ c + \pi \ln(  n w_{t+1} z' )$ where $\pi >0$. We avoid this generality to economize on notation.} \cite*{kremer02} offer a micro-foundation for their fertility curve via quasi-linear parental preferences.

Parents spend their income on consumption, on raising their children, and on the education of their children. Schooling is provided by teachers. We assume that raising a child costs a fraction $\phi \in (0,1)$ of a parent's gross labour income. The budget constraint of a parent with capital $z$ in generation $t$ is therefore given by 
\begin{equation}\label{eq:private_budget}
  c + e n w_t \bar{z}_t + w_t z \phi n \leq w_t z ,
\end{equation} 
where $e$ is schooling time per child and $\bar{z}_t$ denotes the average capital in generation $t$. We assume that the average capital of teachers is the same as the average capital in the population. The cost of education per child is therefore $e w_t \bar{z}_t$. The core components of the budget constraint \eqref{eq:private_budget} are identical to those in \cite*{delacroix03}.

The capital of each child depends on their schooling time, the capital of their parent, and on the average capital in the population, according to\footnote{This functional form specification follows \cite*{Glomm1992Public}, and is not substantively different from the one appearing in \cite*{delacroix03}.} 
\begin{equation}\label{eq:private_transmission}
z' = \theta_t e^\alpha z^\beta \bar{z}_t^\gamma ,
\end{equation}
where $\alpha,\beta \in (0,1)$, $\gamma \in (0,1-\beta]$, and the technology parameter increases deterministically according to
$$\theta_t = \theta (1 + \varrho )^{(1 - \beta - \gamma)t} ,$$
where $\theta>0$. Similarly to \cite*{rangazas2000schooling} and \cite*{delacroix03}, we will obtain endogenous growth for $\beta + \gamma = 1$ and exogenous growth for $\beta + \gamma < 1$.

\subsection{Discussion}
The functional forms that were chosen in \cite*{delacroix03} imply that, in the absence of exogenous stochastic shocks, their model cannot generate atomless steady states along their balanced growth path.\footnote{It is possible to generate long-run inequality in their model, but only through the introduction of exogenous stochastic shocks, with inequality being sustained purely by these shocks. \cite*{bosi11} and \cite*{cordoba16} discuss the difficulties in generating non-degenerate steady state distributions and intergenerational persistence in models with endogenous fertility \`{a} la \cite*{barro89}.} We therefore cannot use the \emph{exact} model of \cite*{delacroix03} to showcase an application of Theorem \ref{thm} that would result in convergence to an atomless steady state. That said, our only substantive departure from \cite*{delacroix03} is in the specification of the parental utility function: they have a Cobb-Douglas specification whereas ours is quasi-linear. In other respects, the core components of our model and of theirs are identical and any remaining differences between the models are minor. For example, the model of \cite*{delacroix03} features old-age consumption and an equation of motion for physical capital. These could also be included in our model without qualitatively affecting the results, but since our aim is only to demonstrate the applicability of the results of Section \ref{sec:results} to a model with endogenous fertility and transmission, we opted for the simplest and most transparent formulation of the model.

\subsection{Optimization}\label{sec:optim}
We normalize the wage per effective hour of labour to unity. The full optimization problem for a parent with capital $z$ in generation $t$ is therefore:
\begin{equation*}
\max_{n\geq0,c\geq0,e\geq0} \;  c +  \ln(n z')  \text{ subject to } c+ n e \bar{z}_t + z\phi n  \leq z \text{ and } z' = \theta_t e^\alpha z^\beta \bar{z}_t^\gamma  .
\end{equation*}
As shown in the appendix, the optimum is given by:
\begin{align*}
n(z) &=
\begin{cases}
 \frac{1-\alpha}{\phi z} & \text{ if } z \geq 1 \\
 \frac{1-\alpha}{\phi } & \text{ otherwise}
\end{cases} ,& 
c(z) &=
\begin{cases}
  z  - 1 & \text{ if } z \geq 1 \\
  0 & \text{ otherwise}
\end{cases} , &
e_t(z) &=\frac{\alpha}{1-\alpha} \phi \frac{z}{\bar{z}_t} .
\end{align*}
This solution is illustrated in Figure \ref{fig:priv}. We have a decreasing relationship between fertility and capital and an increasing relationship between consumption and capital and between schooling time per child and capital.

\begin{figure}
    \centering
    \includegraphics[width=0.5\linewidth]{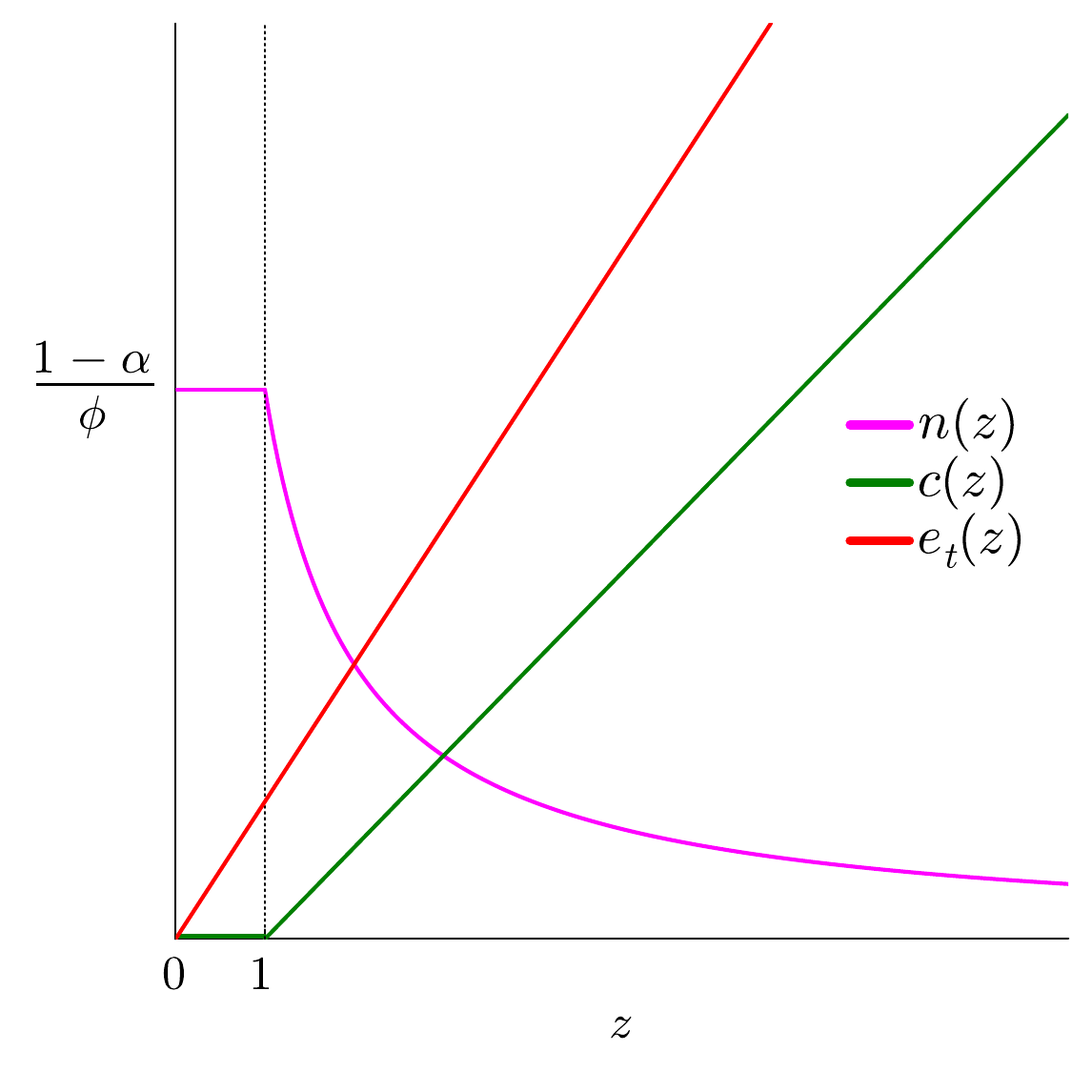}
    \caption{Solution of the parental optimization problem.}
    \label{fig:priv}
\end{figure}

Substituting the optimal solution $e_t$ into \eqref{eq:private_transmission} gives us the endogenous transmission function
\begin{equation}\label{eq:priv_trans}
\tau_t(z) = \zeta_t z^{\alpha + \beta} \bar{z}_t^{\gamma - \alpha} ,
\end{equation}
where $\zeta_t := \theta_t \left( \frac{\alpha}{1-\alpha} \phi  \right)^\alpha$. This endogenous transmission function is simpler than, yet qualitatively similar to, the one derived in \cite*{Croix2004Public}.

Define $\underline{z}_t$ to be the infimum of the support of the distribution of capital in generation $t$. The transmission function \eqref{eq:priv_trans} is a bijection from $\mathcal{Z}_t:= [\underline{z}_t,\infty)$ to $\mathcal{Z}_{t+1}:= [\underline{z}_{t+1},\infty)$ where $\underline{z}_{t+1}:= \tau_t(\underline{z}_t)$. The distribution of capital thus has a time-varying support.

Throughout Section \ref{sec:endo} we assume that, in each generation $t$, $\underline{z}_t \geq 1$. This restriction is easy to satisfy,\footnote{For example, set $\underline{z}_0 \geq 1$ and impose $\theta_t \geq \left(\frac{1-\alpha}{\phi \alpha}\right)^\alpha$ for all $t \geq 0$.} and ensures that all solutions are interior. In particular, under this assumption, for every $t$ the optimal fertility function evaluated over $\mathcal{Z}_t$ does not have a kink.

\subsection{Distributional outcomes}\label{sec:distro}
To make the results of Section \ref{sec:results} applicable to this system, we first recast it into one with a time-invariant domain. Define \defn{relative capital} as
\[
x_t := \frac{z}{\underline{z}_t},
\]
which is in $\mathcal{X}:=[1,\infty)$ for each $z \in \mathcal{Z}_t$.

The distribution of relative capital then evolves, for any $w \in \mathcal{X}$, according to\footnote{If $F_t$ is the distribution of capital at $t$, the distribution of relative capital at $t$ satisfies $\hat{F}_t(z_t) = F_t(x_t \underline{z}_t)$.}
\begin{equation}\label{main_nonlin_rel}
\hat{F}_{t+1}(w) =  \frac{1}{\mathbb{E}_t[\hat{n}_t]}\int_{\mathcal{X}} \mathbbm{1}\left[\hat{\tau}_t(x_t) \leq w \right]  \hat{n}_t(x_t) d\hat{F}_{t}(x_t) ,
\end{equation}
where the fertility and transmission functions over relative capital are obtained by the transformations
\[
\hat{n}_t(x_t) := n(x_t \underline{z}_t) \text{ and } \hat{\tau}_t(x_t) := \frac{\tau_t(x_t \underline{z}_t)}{\underline{z}_{t+1}}.
\]
Evaluating these transformations yields:
\[
\hat{n}_t(x_t) = \frac{1-\alpha}{\phi \underline{z}_t} \frac{1}{x_t} \text{ and } \hat{\tau}_t(x_t) = x_t^{\alpha + \beta} .
\]
The particularly useful feature of having expressed the system in terms of relative capital rather than in terms of capital is that (i) the transmission function $\hat{\tau}_t : \mathcal{X} \rightarrow \mathcal{X}$ is independent of $t$ and satisfies Assumption \ref{as1}, (ii) the fertility function $\hat{n}_t : \mathcal{X} \rightarrow (0,\infty)$ satisfies Assumption \ref{as2} and depends on $t$ only through a scaling factor which, as discussed in Remark \ref{rem:scale}, does not affect the distributional dynamics. The results of Section \ref{sec:results} can therefore be applied to the population process \eqref{main_nonlin_rel} for relative capital.

\begin{remark}
Dividing by $\underline{z}_t$ in the definition relative capital is what works in our system to give us the time invariance of $\hat{\tau}_t$. In a different system, one may need to normalize by a different moment of the distribution of capital.
\end{remark}

The following result is a direct consequence of Theorem \ref{thm} and Proposition \ref{prop:int2}, and we therefore state it without proof.

\begin{theorem}\label{thm:endo}\ 
\begin{enumerate}[leftmargin=0.7cm, label=(\roman*)]
\item If $\alpha+\beta>1$ then there is a steady state distribution of capital that is \emph{Pareto}$[1,1/(\alpha+\beta-1)]$. If, moreover, the initial distribution of relative capital $\hat{F}_0$ admits a continuous and bounded density $\hat{f}_0$ that is supported on $\mathcal{X}$, then the distribution of relative capital converges to a \emph{Pareto}$[1,1/(\alpha+\beta-1)]$ distribution.

\item A degenerate distribution at 1 is a steady state. And if $\alpha+\beta<1$ then, if the initial distribution of relative capital $\hat{F}_{0}$ admits a continuous and bounded density $\hat{f}_{0}$ that is supported on $\mathcal{X}$ then the distribution of relative capital converges to a degenerate distribution at 1.
\end{enumerate}
\end{theorem}

As can be seen above, the value of $\alpha + \beta$ is the key determinant of whether the long-run distribution of relative capital exhibits inequality (the Pareto steady state) or none (the degenerate steady state at 1). In calibrations, it appears customary to assume $\alpha + \beta < 1$, with a typical reason being that if $\alpha + \beta >1$ then the system would exhibit explosive dynamics under non-differential fertility (similar to case (a) in Example A).\footnote{`For individual dynamics to be stable [$\beta$] must not exceed the upper bound [...]' \cite*[p. 1100]{delacroix03}. Since this parameter does not affect choices in their model, stability of the system is one of the main considerations in the calibration of \cite*{delacroix03}.} But this reason is not valid in the presence of differential fertility: even when $\alpha+\beta>1$, so that the transmission function pushes capital towards ever higher levels, a steady state distribution may still exist due to the lower fertility at higher levels of capital. In fact, the empirical estimates in the literature for $\alpha$ (the elasticity of children's capital to educational investment) and $\beta$ (the elasticity of children's capital on parental capital) allow for the possibility that $\alpha+\beta >1$,\footnote{\cite*{delacroix03} report estimates in the literature of $\alpha \in [0.4, 0.8]$ and $\beta \in [0.085, 0.24]$. The calibration in \cite*{cavalcanti2017growth} yields $\beta = 0.55$.} and, as shown in Theorem \ref{thm:endo}, when combined with differential fertility, this case yields the existence of a Pareto steady state to which our system globally converges.

\begin{remark}
We do not claim that `capital', as it is interpreted under this model, has an empirical distribution that is in fact Pareto. Our aim is simply to show that minor modifications to a standard model of endogenous growth with differential fertility can give rise to a system whose outcomes are straightforwardly determined by an application of Theorem \ref{thm}.\footnote{A more complex model of endogenous fertility may give rise to Pareto tails without the full distribution being Pareto, but specifying such a model is beyond our scope. Our aim is to showcase the applicability of Theorem \ref{thm} rather than to present a new full-blown model of endogenous growth.}
\end{remark}

In addition to working out the distributional outcomes, we also derived results on the growth rates of aggregate variables like average capital and population size. These are in the appendix.

\section{Discussion}\label{sec:final}

We analyzed a canonical and deterministic model of the intergenerational transmission of capital that features differential fertility and we identified sufficient conditions for the distribution of capital to converge to degenerate or to atomless steady states. We showed that the sufficient conditions are easy to verify and are satisfied in several parametric examples and in a micro-founded model of endogenous growth. As we have argued, our results also have implications for empirical calibrations. Finally, our comparative statics results provide some insight into the relationship between differential fertility and long-run cross-sectional inequality.

As discussed throughout the text, our investigation has been of the dynamics of a deterministic model which allows us to attribute any long-run inequality to the shapes of the fertility and transmission functions. But we have not discussed an important component of intergenerational transmission, namely, \emph{social mobility}---which is related to but distinct from cross-sectional inequality---and a natural way to include social mobility in our model would be to introduce a stochastic element in the transmission function. Doing so would likely provide further insights into the theoretical relationship between differential fertility, long-run inequality, and social mobility, which we leave for future work.

\clearpage
\appendix
\section{Appendix}

\subsection{Proofs for Section \ref{sec:densities}}
\label{app:thm}
\begin{proof}[Proof of Proposition \ref{prop:density_dynamic}]
From equation \eqref{main_nonlin} we have that for any $\epsilon >0$ and $[x,x+\epsilon] \subseteq \mathcal{X}$,
\begin{align*}
F_{t+1}(x+ \epsilon) - F_{t+1}(x) &= \frac{1}{\mathbb{E}_{t}[n]} \int_{\mathcal{X}} \mathbbm{1}[x \leq \tau(z) \leq x+\epsilon] n(z) dF_{t}(z) \\
&= \frac{1}{\mathbb{E}_{t}[n]} \int_{\rho(x)}^{\rho(x+\epsilon)} n(z) dF_{t}(z) \\
&= \frac{1}{\mathbb{E}_{t}[n]} n(y_\epsilon)  \int_{\rho(x)}^{\rho(x+\epsilon)} dF_{t}(z) \\
&= \frac{1}{\mathbb{E}_{t}[n]} n(y_\epsilon) \left[ F_{t}(\rho(x+\epsilon)) - F_{t}(\rho(x)) \right]
\end{align*}
The step that introduces the quantity $y_\epsilon \in [\rho(x), \rho(x+\epsilon)]$ follows from the mean value theorem for integrals.\footnote{If $g$ is continuous and $h$ is integrable and does not change sign on $[a,b]$, then there is a $c\in[a,b]$ such that $\int_a^b g(x) h(x)dx = g(c) \int_a^b h(x)dx$.} Since $n$ is positive and bounded, we have that $0<\mathbb{E}_{t}[n]<\infty$ for each $t$. Now, dividing both sides by $\epsilon$ and taking the limit $\epsilon \rightarrow 0$ yields
\[
f_{t+1}(x) = \frac{1}{\mathbb{E}_{t}[n]}  n(\rho(x)) \rho'(x) f_{t}(\rho(x)),
\]
which is equivalent to  \eqref{eq:evolution} because $\rho'(x) = 1/\tau'(\rho(x))$.
\end{proof}

\begin{proof}[Proof of Proposition \ref{prop:int1}]
Fix some $y \in \mathcal{X}$. By the transformation of variables $z = \rho(x)$,
\[
\int_{\mathcal X} G_1(x,y) dx = \int_{\mathcal X} \frac{n(\rho(x))}{\tau'(\rho(x))} \frac{n(y)}{\tau'(y)} dx = \frac{n(y)}{\tau'(y)} \int_{\mathcal X} n(z) dz,
\]
so integrability of $n(\cdot)$ implies integrability of $G_t(\cdot,y)$ for some $t \geq 1$.
\end{proof}

\begin{proof}[Proof of Proposition \ref{prop:int2}]
For $x \in \mathcal{X}=[1,\infty)$, and fixing some $y \in \mathcal{X}$, define the function
$$ \tilde{G}_t(x,y) := \prod_{i=1}^t \frac{\tilde{n}(\rho^{[i]}(x))}{\tau'(\rho^{[i]}(x))} \frac{\tau'(y)}{\tilde{n}(y)} \propto \prod_{i=1}^t x^{\frac{1-m-a}{a^i}} = x^{(1-m-a)\frac{1-1/a^t}{a-1}} ,$$
where $\tilde{n}(x) := x^{-m}$.

We now show that there is a $t \geq 1$ sufficiently large such that $\tilde{G}_t(\cdot,y)$ is integrable. We then employ the limit comparison test to show that the integrability of $\tilde{G}_t(\cdot,y)$ implies the integrability of $G_t(\cdot,y)$.\footnote{The limit comparison test for improper integrals states that if $g$ and $h$ are continuous positive functions on $[c,\infty)$ and $\lim_{x \rightarrow \infty} g(x)/h(x) = L \in (0,\infty)$ then $\int_c^\infty g(x) dx$ and $\int_c^\infty h(x) dx$ either both converge or both diverge. Furthermore, if $\lim_{x \rightarrow \infty} g(x)/h(x) = 0$ then convergence of $\int_c^\infty h(x) dx$ implies convergence of $\int_c^\infty g(x) dx$.}

For any $m > 0$, the function $\tilde{G}_t(\cdot,y)$ is integrable whenever
$$ (1-m-a)\frac{1-1/a^t}{a-1} < -1 .$$
Re-arranging yields $t> \ln\left( \frac{m+a-1}{m}\right)/\ln(a)$. So define $T:= \ceil*{ \ln\left( \frac{m+a-1}{m}\right) / \ln(a) }$ and consider the ratio of the two functions of interest at any $t \geq T$. We obtain:
\begin{align*}
\lim_{x \rightarrow \infty} \frac{G_t(x,y)}{\tilde{G}_t(x,y)} &= \lim_{x \rightarrow \infty} \prod_{i=1}^{t} \frac{n(\rho^{[i]}(x))}{\tau'(\rho^{[i]}(x))} \frac{\tau'(y)}{n(y)} \frac{\tau'(\rho^{[i]}(x))}{\tilde{n}(\rho^{[i]}(x))} \frac{\tilde{n}(y)}{\tau'(y)} \\
&= \left( \frac{\tilde{n}(y)}{n(y)} \right)^t \times \prod_{i=1}^{t}   \lim_{x \rightarrow \infty } \frac{n(\rho^{[i]}(x))}{\tilde{n}(\rho^{[i]}(x))} \\
&= \left( \frac{\tilde{n}(y)}{n(y)} \right)^t \times M^t < \infty \, .
\end{align*}
The key step follows from our assumption that for some $m>0$, $\lim_{x \rightarrow \infty} x^m n(x) = M$ where $0\leq M <\infty$. The above suffices to show integrability of $G_t(\cdot,y)$ for some $t \geq 1$.
\end{proof}

\subsubsection{Proof of Theorem \ref{thm}}
The proof of Theorem \ref{thm} is done in three steps. First, we show that the population process can be decomposed into $K+1$ parallel processes, one for each interval $\mathcal{X}_k$, and in Proposition \ref{prop:recover} we show that the dynamics of the population process can be recovered from the dynamics of each of the parallel processes. Second, under the assumption of well-behavedness, Proposition \ref{prop:parallel} shows that each parallel process has a steady state (either degenerate or atomless) and converges to such a steady state. The proof of Proposition \ref{prop:parallel} is the most involved and is built on several lemmas. Finally we show that, when combined with genericity, Proposition \ref{prop:parallel} yields Theorem \ref{thm}.

Let's start with Proposition \ref{prop:recover}. Suppose that $F_0$ admits a density $f_0$. For each $k \in \{0,\dots,K\}$, construct a density $f_{k,t}$ as follows: for each $x \in \mathcal{X}_k$ set $f_{k,0}(x) := f_0(x)$ and for $t \geq 0$, let
\begin{equation}\label{eq:evolution_k}
f_{k,t+1}(x) := \frac{1}{\mathbb{E}_{k,t}[n]} \frac{n(\rho(x))}{\tau'(\rho(x))} f_{k,t}(\rho(x))  
\end{equation}
where $\mathbb{E}_{k,t}[n]:= \int_{\mathcal{X}_k} n(z) f_{k,t}(z) dz$. For each $x \not\in \mathcal{X}_k$ and each $t\geq 0$, set $f_{k,t}(x):=0$. The evolution of the population process \eqref{main_nonlin} can be recovered from the $K+1$ parallel processes defined by \eqref{eq:evolution_k} thanks to the following result.
\begin{proposition}\label{prop:recover}
For each $k\in \{0,\dots,K\}$, $x \in \mathcal{X}_k$, and $t \geq 1$,
\begin{equation}\label{eq:recover}
 f_{t}(x) = w_{k,t} f_{k,t}(x)  ,
\end{equation}
where
\begin{equation}\label{eq:w}
    w_{k,t} := \frac{\prod_{i=0}^{t-1} \mathbb{E}_{k,i}[n]}{\sum_{k'=0}^K \prod_{i=0}^{t-1} \mathbb{E}_{k',i}[n]} .
\end{equation}
\end{proposition}
\begin{remark}
Observe that in \eqref{eq:evolution_k} the evolution of the density $f_{k,t}:\mathcal{X}_k \to [0,\infty)$, which is simply $f_t$ truncated to $\mathcal{X}_k$, is independent of the evolution of $f_{k',t}$ for any $k \neq k'$.
\end{remark}
\begin{remark}\label{cor:growth2}
If $F_k^*$ is a steady state of the process $F_{k,t} \mapsto F_{k,t+1}$ that admits a density $f_k^*$ for which $f_k^*(s)>0$ and $s$ is a fixed point of $\tau$ then $\mathbb{E}_k^*[n] = n(s)/\tau'(s)$ is the steady state population growth factor of this $k$th parallel process. The result follows directly from evaluating \eqref{eq:evolution_k} in steady state at $s$.
\end{remark}
\begin{proof}[Proof of Proposition \ref{prop:recover}]
First, observe that \eqref{eq:w} solves the system 
\[
w_{k,t+1} = \frac{w_{k,t} \mathbb{E}_{k,t}[n]}{\sum_{k'=0}^K w_{k',t} \mathbb{E}_{k',t}[n]}
\]
where $w_{k,0}:=1$ for each $k$. We therefore only need to show that \eqref{eq:recover} holds for $w_{k,t}$ satisfying the equation above. We do this by induction. By construction, $f_0(x) = w_{k,0} f_{k,0}(x)$ for each $x \in \mathcal{X}_k$ where $w_{k,0}:=1$. For an induction, suppose $f_t(x) = w_{k,t} f_{k,t}(x)$ for each $x \in \mathcal{X}_k$. By \eqref{eq:evolution} and \eqref{eq:evolution_k} we have that for any $x \in \mathcal{X}_k$,
\[
f_{t+1}(x) = \frac{1}{\mathbb{E}_t[n]} \frac{n(\rho(x))}{\tau'(\rho(x))} f_t(\rho(x)) = \frac{w_{k,t}}{\mathbb{E}_t[n]} \frac{n(\rho(x))}{\tau'(\rho(x))} f_{k,t}(\rho(x)) = \frac{w_{k,t} \mathbb{E}_{k,t}[n]}{\mathbb{E}_t[n]}f_{k,t+1}(x).
\]
To complete the proof it now suffices to show that for each $t\geq 0$, $\mathbb{E}_t[n] = \sum_{k'=0}^K w_{k',t} \mathbb{E}_{k',t}[n]$. We again show this by induction on $t$. The base case is true by construction. So let's assume that the statement is true for $t$ and we'll show that it must also be true at $t+1$. From the above, we know that
\begin{align*}
\mathbb{E}_{t+1}[n] = \int_{\mathcal{X}} n(x) f_{t+1}(x) dx &= \sum_{k=0}^K  \int_{\mathcal{X}_k} n(x) f_{t+1}(x) dx \\
&= \sum_{k=0}^K  \int_{\mathcal{X}_k} \frac{w_{k,t} \mathbb{E}_{k,t}[n]}{\mathbb{E}_t[n]}n(x) f_{k,t+1}(x) dx \\
&=\sum_{k=0}^K \frac{w_{k,t} \mathbb{E}_{k,t}[n]}{\mathbb{E}_t[n]} \mathbb{E}_{k,t+1}[n] \\
&=\sum_{k=0}^K w_{k,t+1} \mathbb{E}_{k,t+1}[n]
\end{align*}
The final step follows from the inductive hypothesis and the definition of $w_{k,t+1}$.
\end{proof}
\noindent Proposition \ref{prop:recover} is a useful result. It tells us that, in order to determine the dynamics of the population process \eqref{main_nonlin}, we can simply work out the dynamics of each of the $K+1$ parallel processes separately and then recover the dynamics for the population process from \eqref{eq:recover}.

The next step is to work out the dynamics of each of the parallel processes. Observe that for any $t\geq 1$, and $k \in \{0,\dots,K\}$ and any $x,y \in \mathcal X$, repeated iteration of \eqref{eq:evolution_k} yields:
\begin{equation}\label{eq:convK}
f_{k,t}(x) = \frac{\left(\frac{n(y)}{\tau'(y)}\right)^t}{\prod_{i=0}^{t-1} \mathbb{E}_{k,i}[n]} G_t(x,y) f_0(\rho^{[t]}(x)) \propto G_t(x,y) f_0(\rho^{[t]}(x)) 
\end{equation}
and therefore
\begin{equation}\label{eq:conK2}
\prod_{i=0}^{t} \mathbb{E}_{k,i}[n] = \left(\frac{n(y)}{\tau'(y)}\right)^t \int_{\mathcal{X}_k} n(x) G_t(x,y) f_0(\rho^{[t]}(x)) dx 
\end{equation}
Proposition \ref{prop:parallel} below establishes conditions under which each of the $K+1$ processes has a steady state to which it converges (which primarily consists in showing that the limit of equation \eqref{eq:convK} exists and is integrable). 
\begin{proposition}\label{prop:parallel}
Consider an interval $\mathcal{X}_k$ for some $k \in \{0,\dots,K\}$ and  some $s \in \mathcal{X}_k$ where $s$ is a fixed point of $\tau$. Assume that $G_t(\cdot,s)$ is integrable on $\mathcal{X}$ for some $t \geq 1$.
\begin{enumerate}[leftmargin=0.7cm, label=(\roman*)]
    \item If $s$ is a source, and the function $n(\cdot)/\tau'(\cdot)$ is well-behaved at $s$ and endpoint maximal at $s$ in $\mathcal{X}_k$, then there exists a steady state supported on $\mathcal{X}_{k}$ with density
\[
f_k^*(x) \propto \lim_{t \to \infty} G_t(x,s) .
\]
If, moreover, the initial distribution $F_{0}$ admits a continuous and bounded density $f_{0}$ that is supported on $\mathcal X$, then for all $x \in \mathcal{X}_k$, $\lim_{t \to \infty} f_{k,t}(x) = f_k^*(x)$, and there is a constant $c_k>0$ such that 
\begin{equation}\label{eq:prodlim1}
\lim_{t \to \infty }\left(\frac{\tau'(s)}{n(s)}\right)^t \prod_{i=0}^{t} \mathbb{E}_{k,i}[n] = c_k .
\end{equation}
\item If, instead, $s$ is a sink, and the function $n(\cdot)$ is well-behaved at $s$ and endpoint maximal at $s$ in $\mathcal{X}_k$, and the initial distribution $F_{0}$ admits a continuous and bounded density $f_{0}$ that is supported on $\mathcal X$, then for all $x \in \mathcal{X}_k$, $\lim_{t \to \infty} F_{k,t} (x) = \mathbbm{1}[s \leq x]$, and there is a constant $c_k>0$ such that
\begin{equation}\label{eq:prodlim2}
\lim_{t \to \infty }\left(\frac{1}{n(s)}\right)^t \prod_{i=0}^{t} \mathbb{E}_{k,i}[n] = c_k .
\end{equation}
\end{enumerate}
\end{proposition}
Next, we introduce two functions, $\kappa$ and $q$, and a functional $Q$, and prove a number of lemmas regarding these mappings. The proof of Proposition \ref{prop:parallel} then makes use of these lemmas with appropriate substitutions for $\kappa$ and $q$.

Throughout the rest of Appendix \ref{app:thm} let $\kappa: \mathcal{X} \rightarrow \mathcal{X}$ denote a continuously differentiable and strictly increasing function. Moreover, assume that $\kappa$ has exactly $1 \leq K <\infty$ fixed points in $\mathcal{X}$ and that these fixed points are the same as those of the transmission function $\tau$. We therefore denote the fixed points of $\kappa$ by $s_1,\dots,s_K$ with $s_1 < \cdots < s_K$. Let $q : \mathcal{X} \rightarrow (0,\infty)$ denote a positive, bounded, and continuously differentiable function. Finally, for any $x,y\in \mathcal{X}$, any integer $m \geq 0$, and any $t \geq m$, define the functional
\[
Q_{m}^t[q,\kappa]( x, y ) := \prod_{i=m}^t \frac{q(\kappa^{[i]}(x))}{q(y)} .
\]

\begin{remark}\label{rem:h}
For any continuously differentiable function $h: \mathcal{X} \to (0,\infty)$, if $h'(x)>0$ for some $x \in \mathcal{X}$ then there is $\epsilon>0$ such that $h(y)>h(x)$ for all $y \in (x,x+\epsilon)$, and if $h'(x)<0$ for some $x \in \mathcal{X}$ then there is $\epsilon>0$ such that $h(y)<h(x)$ for all $y \in (x,x+\epsilon)$.\footnote{The proof is straightforward. If a function $\mu$ is continuous at $x$ and $\mu(c)>0$ there there is $\epsilon>0$ such that $\mu(x)>0$ for all $x \in (c-\epsilon,c+\epsilon)$. So, since $h'(x)>0$ and $h'$ is continuous, there must exist some $\epsilon>0$ such that $h'(y)>0$ for all $y \in (x,x+\epsilon)$. This now implies the result. If not, then there would be some $y \in (x,x+\epsilon)$ such that $h(y) \leq h(x)$ but, by the mean value theorem, there would then exists some $c (x,y)$ such that $h'(c) = (h(y)-h(x))/(y-x) \leq 0$; a contradiction.}
\end{remark}

\begin{lemma}\label{lem:kappa1}
Consider an interval $\mathcal{X}_k$ for some $k \in \{0,\dots,K\}$. Suppose $s \in \mathcal{X}_k$ is a fixed point of $\kappa$ and that $\kappa'(s)<1$. If $s$ is the left (respectively, right) endpoint of $\mathcal{X}_k$ then for any $x \in \interior \mathcal{X}_k$ the sequence $\kappa^{[1]}(x), \kappa^{[2]}(x),\dots$ is strictly decreasing (respectively, strictly increasing) and $\lim_{t\rightarrow \infty} \kappa^{[t]}(x)=s$.    
\end{lemma}
\begin{proof}[Proof of Lemma \ref{lem:kappa1}]
We give the proof for the case in which $s$ is the left endpoint of $\mathcal{X}_k$ since the case in which it is the right endpoint is entirely analogous.

Since $\kappa$ is continuously differentiable and $1-\kappa'(s)>0$, by Remark \ref{rem:h}, we have that $x-\kappa(x)>0$ for all $x \in (s,s+\epsilon)$. And, since $\kappa$ is continuous, this implies that $x-\kappa(x)>0$ for all $x \in \interior \mathcal{X}_k$. Since $\kappa$ is strictly increasing, $x>\kappa(x)$ implies $\kappa(x)>\kappa(\kappa(x))$, so we conclude that the sequence $\kappa^{[1]}(x),\kappa^{[2]}(x),...$ is strictly decreasing for all $x \in \interior \mathcal X$. The sequence is bounded below by $s$ and thus convergent and its limit must be $s$.  
\end{proof}

\begin{lemma}\label{lem:kappa2}
Consider an interval $\mathcal{X}_k$ for some $k \in \{0,\dots,K\}$. Suppose $s \in \mathcal{X}_k$ is a fixed point of $\kappa$ and that $\kappa'(s)<1$. Consider $x \in \interior \mathcal{X}_k$. If $s$ is the left endpoint of $\mathcal{X}_k$ then for any $\epsilon>0$ there is a positive and finite integer $ T $ such for all $y \in [s,x]$, $\kappa^{[t]}(y) \in [s,s+\epsilon)$ for all $t \geq T$. Similarly, if $s$ is the right endpoint of $\mathcal{X}_k$ then for any $\epsilon>0$ there is a positive and finite integer $ T $ such for all $y \in [x,s]$, $\kappa^{[t]}(y) \in (s-\epsilon,s]$ for all $t \geq T$.
\end{lemma}
\begin{proof}[Proof of Lemma \ref{lem:kappa2}]
Again, we prove only the case in which $s$ is the left endpoint of $\mathcal{X}_k$ as the alternative case is analogous.

Let $x \in \interior \mathcal{X}_k$ and consider some $y \in [s,x]$. We know from Lemma \ref{lem:kappa1} that the sequence $\kappa^{[1]}(y), \kappa^{[2]}(y),...$ is strictly decreasing so we can fix $\epsilon>0$ and define $T(y)$ to be the smallest index $T(y) \geq 1$ such that $\kappa^{[t]}(y) \in [s, s + \epsilon)$ for all $t \geq T(y)$. Then define $T := \sup_{y \in [s,x]} T(y) = T(x)$, which exists and is finite.
\end{proof}

\begin{lemma}\label{lem:q}
Consider an interval $\mathcal{X}_k$ for some $k \in \{0,\dots,K\}$. Suppose $s \in \mathcal{X}_k$ is a fixed point of $\kappa$ and that $s$ is not a stationary point of $q$. If $s$ is the left endpoint of $\mathcal{X}_k$ then there is $\epsilon>0$ such that $q(x) - q(s)$ is positive for all $x \in (s,s+\epsilon)$ or negative for all $x \in (s,s+\epsilon)$. Similarly, if $s$ is the right endpoint of $\mathcal{X}_k$ then there is $\epsilon>0$ such that $q(x) - q(s)$ is positive for all $x \in (s-\epsilon,s)$ or negative for all $x \in (s-\epsilon,s)$.
\end{lemma}
\begin{proof}[Proof of Lemma \ref{lem:q}]
We prove only the case in which $s$ is the left endpoint of $\mathcal{X}_k$ as the alternative case is analogous.

If $q'(s)>0$ then, by Remark \ref{rem:h}, $q(x)> q(s)$ for all $x \in (s,s+\epsilon)$. An analogous argument applies to $q'(s)<0$.
\end{proof}

\begin{lemma} \label{lem:Q}
Consider an interval $\mathcal{X}_k$ for some $k \in \{0,\dots,K\}$. Suppose $s \in \mathcal{X}_k$ is a fixed point of $\kappa$. If $s$ is not a stationary point of $q$ and $\kappa'(s)<1$ then 
\begin{equation} \label{infproduct}
\lim_{t \rightarrow \infty} Q_m^t[q,\kappa]( x, s)
\end{equation}
converges to a positive real number for each $x \in \{s\} \cup \interior \mathcal{X}_k$.
\end{lemma}
\begin{proof}[Proof of Lemma \ref{lem:Q}]
Convergence of \eqref{infproduct} at $x = s$ itself is obvious since all factors in the product are equal to 1. Now consider $x \in \interior \mathcal{X}_k$. Convergence of \eqref{infproduct} to a positive real number is equivalent to convergence of $\sum_{i=m}^\infty  r_i$ where 
\[
r_i := \frac{q(x_i) - q(s)}{q(s)} 
\]
and $x_i := \kappa^{[i]}(x)$. By the ratio test, the sum converges if $\lim_{i \rightarrow \infty} \left| r_{i+1}/r_i\right|<1$. We have that
\[
\lim_{i \rightarrow \infty} \left| \frac{r_{i+1}}{r_i}\right| = \lim_{i \rightarrow \infty} \left|\frac{q(\kappa(x_{i})) - q(s)}{q(x_i) - q(s)} \right| = \lim_{x \rightarrow s} \left|\frac{q(\kappa(x)) - q(s)}{q(x) - q(s)} \right| ,
\]
where the final step follows from Lemma \ref{lem:kappa1}. By l'Hopital's rule we have\footnote{We employ a particularly simple form of the rule: if $g$ and $h$ are continuously differentiable on an interval containing $c$, $g(c)=h(c)=0$, and $h'(c)\neq 0$, then $\lim_{x \rightarrow s} g(x)/h(x) = \lim_{x \rightarrow s} g'(x)/h'(x) = g'(c)/h'(c)$.}
\[
\lim_{x \rightarrow s} \left|\frac{q(\kappa(x)) - q(s)}{q(x) - q(s)} \right| = \lim_{x \rightarrow s} \left| \frac{q'(\kappa(x)) \kappa'(x)}{q'(x)} \right| = \left| \frac{q'(\kappa(s)) \kappa'(s)}{q'(s)} \right| = \left| \kappa'(x) \right| < 1,
\]
which completes the proof.
\end{proof}

\begin{lemma} \label{lem:D}
Consider an interval $\mathcal{X}_k$ for some $k \in \{0,\dots,K\}$. Suppose $s \in \mathcal{X}_k$ is a fixed point of $\kappa$. If $s$ is not a stationary point of $q$, $q$ is endpoint maximal at $s$ in $\mathcal{X}_k$, and $\kappa'(s)<1$, then for any fixed integer $m\geq1$ we have $1 \leq D_m < \infty$, where
\[
D_m := \sup_{x\in \mathcal{X}_k} \sup_{t \geq m} Q_{m}^t[q,\kappa](x,s).
\]
\end{lemma}

\begin{proof}[Proof of Lemma \ref{lem:D}]
For $x \in \mathcal{X}_k$ define 
\[
\Delta(x) := \sup_{t \geq m} Q_{m}^t[q,\kappa](x,s) .
\]
We know that $D_m \geq 1$ since $\Delta(x)=1$ whenever $x = s$.

It remains for us to show that $D_m < \infty$. If $q(x) \leq q(s)$ for all $x \in \mathcal{X}_k$ then each factor of $\Delta(x)$ is at most 1, so $D_m$ is finite. Assume instead, for the rest of the proof, that $q(x) > q(s)$ for some $x \in \mathcal{X}_k$ and define
\[
\bar{s}:= \sup \{x \in \mathcal{X}_k : q(x) > q(s)\} .
\]
Without loss of generality, let $s$ be the left endpoint of $\mathcal{X}_k$ so that $s=s_k$. Since $q$ is endpoint maximal at $s$ we know that $\bar{s} < s_{k+1}$. Moreover, since $s$ is not a stationary point of $q$, by Lemma \ref{lem:q}, there exists $\epsilon>0$ with the property that $q(x) - q(s)$ is negative all $x \in (s,s+\epsilon)$, or positive for all $x \in (s,s+\epsilon)$.\footnote{\label{fn:tech}Assuming that $s$ is not a stationary point guarantees the existence of such an $\epsilon$ by ruling out pathological cases. For example, if we define $\mu:[0,1] \to \mathbb{R}$ to be the function $\mu(x) = x^3 \sin(1/x)$ for $x>0$ and $\mu(x)=0$ for $x=0$, then $\mu$ is continuously differentiable over its domain but it is not possible to find an $\epsilon>0$ such that $\mu(x) - \mu(0)$ is negative for all $x \in (0,\epsilon)$ or positive for all $x \in (0,\epsilon)$. This is because for any finite interval $(0,\epsilon)$, $\mu$ has an infinite number of zeros within the interval.} Moreover, defining $\underline{s}:=s + \epsilon$ we can choose this $\epsilon$ to be sufficiently small so that $s = s_k < \underline{s} < \bar{s} < s_{k+1}$. We can therefore split the interval $\mathcal{X}_k$ into three parts and write
\[
D_m = \max \{ \sup_{x\in [s,\underline{s}]}\Delta(x), \sup_{x\in [\underline{s},\bar{s}]}\Delta(x), \sup_{x\in [\bar{s},s_{k+1}] }\Delta(x)\}
\]
We now show that the supremum of $\Delta$ over each of the three intervals is finite.

(a) $\sup_{x\in [s,\underline{s}]}\Delta(x)$ is finite. By Lemma \ref{lem:q} we have two cases. Either $q(x) \leq q(s)$ for each $x \in [s,\underline{s}]$ in which case each factor in the expression for $\Delta(x)$ is at most 1 and so $\sup_{x\in [s,\underline{s}]}\Delta(x)$ is finite. Or $q(x) > q(s)$ for each $x \in (s,\underline{s}]$, so each factor in the expression for $\Delta(x)$ exceeds 1. In this case, for each $x \in [s,\underline{s}]$,
\[
\Delta(x) = \sup_{t\geq m}\prod_{i=m}^t \frac{q(\kappa^{[i]}(x))}{q(s)} = \lim_{t \rightarrow \infty} \prod_{i=m}^t \frac{q(\kappa^{[i]}(x))}{q( s )},
\]
but by Lemma \ref{lem:Q} the limit converges to a positive real number and is therefore finite.

(b) $\sup_{x\in [\underline{s},\bar{s}]}\Delta(x)$ is finite. By Lemma \ref{lem:kappa2} there is an finite integer $T$ such that for all $x \in [\underline{s},\bar{s}]$, $\kappa^{[i]}(x) \in [s,\underline{s}]$ for all $i \geq T$. We can therefore write
\begin{align}
\sup_{x\in [\underline{s},\bar{s}]}\Delta(x) &= \sup_{x\in [\underline{s},\bar{s}]} \sup_{t \geq m} \left\{ \left( \prod_{i=m}^{T} \frac{q(\kappa^{[i]}(x))}{q( s )} \right) \times \left( \prod_{i=T+1}^t \frac{q(\kappa^{[i]}(x))}{q( s )} \right) \right\} \nonumber \\
&\leq  M \times \sup_{x\in [\underline{s},\bar{s}]} \sup_{t>m}  \prod_{i=T+1}^t \frac{q(\kappa^{[i]}(x))}{q( s )} \label{eq:finite_prod}
\end{align}
The second line follows from the fact that we can find some $M<\infty$ to bound the first product because the product runs over a finite number of indices and each factor is bounded due to the boundedness of $q$. Each $\kappa^{[i]}(x)$ in \eqref{eq:finite_prod} belongs to the interval $[s,\underline{s}]$ and therefore \eqref{eq:finite_prod} is finite by the argument (a) given above.

(c) $\sup_{x\in [\bar{s},s_{k+1}]}\Delta(x)$ is finite. All factors $q(y)/q(s)$ in the expression for $\Delta(x)$ such that $y \in [\bar{s},s_{k+1}]$ are at most 1, so $\sup_{x\in [\bar{s},s_{k+1}] }\Delta(x) \leq \sup_{x\in [\underline{s},\bar{s}]}\Delta(x)$, and the latter is finite by argument (b) above.
\end{proof}

\begin{lemma}\label{lem:int}
Consider an interval $\mathcal{X}_k$ for some $k \in \{0,\dots,K\}$. Suppose $s \in \mathcal{X}_k$ is a fixed point of $\kappa$. If $s$ is not a stationary point of $q$, $q$ is endpoint maximal at $s$ in $\mathcal{X}_k$, $\kappa'(s)<1$, and $Q_m^t[q,\kappa](\cdot,s)$ is integrable on $\mathcal{X}_k$ for some $t \geq m$, then for any positive, continuous, and bounded function $\psi : \mathcal{X} \rightarrow (0,M)$ where $M < \infty$,
\[
\lim_{t \rightarrow \infty} \int_{\mathcal{X}_k} \psi(\kappa^{[t]}(x)) Q_m^t[q,\kappa](x,s) dx = \int_{\mathcal{X}_k} \psi(s) \lim_{t \rightarrow \infty} Q_m^t[q,\kappa](x,s) dx  .
\]
\end{lemma}
\begin{proof}[Proof of Lemma \ref{lem:int}]
From Lemma \ref{lem:kappa1} and continuity of $\psi$, we have that $\lim_{t \rightarrow \infty} \psi(\kappa^{[t]}(x)) = \psi(s)$ for all $x \in \interior \mathcal{X}_t$. Lemma \ref{lem:Q} guarantees that $\lim_{t \rightarrow \infty} Q_m^t[q,\kappa](x,s)$ exists and is positive for each $x \in \interior \mathcal{X}_t$.

Suppose $Q_m^{T}[q,\kappa](\cdot,s)$ is integrable on $\mathcal{X}_t$ some $T \geq m$. For all $t > T$, we can write
\begin{align*}
\psi(\kappa^{[t]}(x)) \, Q_m^t[q,\kappa]( x, s) &= \psi(\kappa^{[t]}(x)) \, Q_m^T[q,\kappa](x,s) \prod_{i=T +1}^t \frac{q(\kappa^{[i]}(x))}{q( s )} \\
&\leq M \, Q_m^{T}[q,\kappa](x,s) \, D_{T+1} .
\end{align*}
where $D_{T+1} := \sup_{x \in \mathcal{X}_t} \sup_{t \geq T+1} \prod_{i=T +1}^t \frac{q(\kappa^{[i]}(x))}{q( s )}$. By Lemma \ref{lem:D}, we have $1 \leq D_{T+1} < \infty$. It follows that the upper bound $M Q_m^{T}[q,\kappa](x,s)  D_T$ is integrable and pointwise dominates $\psi(\kappa^{[t]}(x)) Q_m^t[q,\kappa]( x, s)$ for all $t \geq T$. By the dominated convergence theorem, we therefore have that $\psi( s ) \lim_{t \rightarrow \infty} Q_m^t[q,\kappa](\cdot,s)$ is integrable and its integral is given by
\[
\lim_{t \rightarrow \infty} \int_{\mathcal{X}_t} \psi(\kappa^{[t]}(x)) Q_m^t[q,\kappa](x,s) dx = \int_{\mathcal{X}_t} \psi( s ) \lim_{t \rightarrow \infty} Q_m^t[q,\kappa](x,s) dx  ,
\]
which completes the proof.
\end{proof}

\begin{proof}[Proof of Proposition \ref{prop:parallel}]
Consider an interval $\mathcal{X}_k$ for some $k \in \{0,\dots,K\}$. Suppose $s \in \mathcal{X}_k$ is a fixed point of the transmission function $\tau$. Observe that since we assume that $G_t(\cdot,s)$ is integrable on $\mathcal{X}$ for some $t \geq 1$, it is also integrable on $\mathcal{X}_k$ for some $t\geq 1$.

Part (i). With the substitutions $\kappa = \rho$ and $q(\cdot) = n(\cdot)/\tau'(\cdot)$ and we have that
\[
G_t(x,s) = Q_1^t[n/\tau',\rho](x,s).
\]
Since $\tau$ satisfies Assumption \ref{as2}, $\rho$ is continuously differentiable and strictly increasing. And since $s$ is a source, $\tau'(s)>1$ so $\rho'(s)<1$.   Moreover, since $\tau$ and $n$ satisfy Assumptions \ref{as1} and \ref{as2} respectively, the function $n(\cdot)/\tau'(\cdot)$ is positive, bounded, and continuously differentiable. With the substitutions $\kappa = \rho$ and $q(\cdot) = n(\cdot)/\tau'(\cdot)$ we therefore satisfy all the conditions of Lemmas \ref{lem:kappa1} through \ref{lem:int}.

By Lemma \ref{lem:Q}, $\lim_{t \rightarrow \infty} G_t(x,s)$ converges for each $x \in \{s\} \cup \interior \mathcal{X}_k$. Moreover, by Lemma \ref{lem:int}, $\lim_{t \rightarrow \infty} G_t(\cdot,s)$ is integrable and its integral is given by
\[
\lim_{t \rightarrow \infty} \int_{\mathcal{X}_k} G_t(x,s) dx = \int_{\mathcal{X}_k} \lim_{t \rightarrow \infty}G_t(x,s) dx .
\]

We first prove that a steady state density on $\mathcal{X}_k$ is proportional to $\lim_{t \rightarrow \infty} G_t(\cdot,s)$. Consider a steady state density $f_k^*$. By equation \eqref{eq:evolution_k} and Remark \ref{cor:growth2}, if $f_k^*(s)>0$, we have that
\[
f_k^*(x) = \frac{1}{\mathbb{E}_k^*[n]}  \frac{n(\rho(x))}{\tau'(\rho(x))} f_k^*(\rho(x)).
\]
Clearly, setting 
\[
f_k^*(x) = \frac{\lim_{t \rightarrow \infty} G_t(x,s)}{\int_{\mathcal{X}_k}\lim_{t \rightarrow \infty} G_t(y,s) dy}
\]
satisfies the above and $f_k^*$ is a density since $\lim_{t \rightarrow \infty} G_t(x,s)$ exists and is integrable.

Now, we turn to convergence. By equation \eqref{eq:convK} {\color{blue}},
\begin{align*}
    \lim_{t \rightarrow \infty} f_{k,t}(x) = \lim_{t \rightarrow \infty} \frac{G_t(x,s) f_{0}(\rho^{[t]}(x))}{\int_{\mathcal{X}_k}G_t(y,s) f_{0}(\rho^{[t]}(y)) dy} =  \frac{\lim_{t \rightarrow \infty} G_t(x,s) f_{0}(s)}{\int_{\mathcal{X}_k} \lim_{t \rightarrow \infty} G_t(y,s) f_{0}(s) dy} = f_k^*(x).
\end{align*}
Pulling the limit into the integral follows Lemma \ref{lem:int}.

Finally, from Lemma \ref{lem:int} again, and from equation \eqref{eq:conK2} we have that
\begin{align*}
\lim_{t \to \infty} \left(\frac{\tau'(s)}{n(s)}\right)^t \prod_{i=0}^{t} \mathbb{E}_{k,i}[n] &= \lim_{t \to \infty} \int_{\mathcal{X}_k} n(x) G_t(x,s) f_0(\rho^{[t]}(x)) dx \\
&=  \int_{\mathcal{X}_k} n(x) \lim_{t \to \infty} G_t(x,s) f_0(s) dx 
\end{align*}
which is positive and finite, thus proving \eqref{eq:prodlim1}.

Part (ii). To prove this part, we make the substitutions $\kappa = \tau$ and $q= n$. With these substitutions,
\[
Q_0^{t-1}[n,\tau](x,s) = \prod_{i=0}^{t-1}  \frac{n(\tau^{[i]}(y))}{n(s)} .
\]
Since $\tau$ satisfies Assumption \ref{as2}, it is continuously differentiable and strictly increasing. And since $s$ is a sink, $\tau'(s)<1$. Moreover, since $n$ satisfies Assumption \ref{as1}, the function is positive, bounded, and continuously differentiable. With the substitutions $\kappa = \tau$ and $q= n$ we therefore satisfy all the conditions of Lemmas \ref{lem:kappa1} through \ref{lem:D}.

By Lemma \ref{lem:Q}, $\lim_{t \rightarrow \infty} Q_0^{t-1}[n,\tau](x,s)$ converges for each $x \in \{s\} \cup \interior \mathcal{X}_k$. Moreover, $Q_0^{t-1}[n,\tau](\cdot,s)$ is integrable over $\mathcal{X}_k$ for some $t\geq 1$ because 
\begin{align}\label{eq:transformation}
\int_{\mathcal{X}_k} G_t(x,s) dx = \int_{\mathcal{X}_k} Q_1^t[n/\tau',\rho](x,s) dx &= \int_{\mathcal{X}_k} \prod_{i=1}^t \frac{n(\rho^{[i]}(x))}{\tau'(\rho^{[i]}(x))} \frac{\tau'(s)}{n(s)} dx \nonumber\\
&=\int_{\mathcal{X}_k} \prod_{i=0}^{t-1} n(\tau^{[i]}(y)) \frac{\tau'(s)}{n(s)} dy \nonumber\\
&= \left( \tau'(s) \right)^t \int_{\mathcal{X}_k} Q_0^{t-1}[n,\tau](x,s) dx
\end{align}
and we assumed that $G_t(\cdot,s)$ is integrable over $\mathcal{X}_k$ for some $t\geq 1$. We therefore also satisfy the conditions of Lemma \ref{lem:int} so $\lim_{t \rightarrow \infty} Q_0^{t-1}[n,\tau](\cdot,s)$ is integrable and its integral is given by
\[
\lim_{t \rightarrow \infty} \int_{\mathcal{X}_k} Q_0^{t-1}[n,\tau](x,s) dx = \int_{\mathcal{X}_k} \lim_{t \rightarrow \infty} Q_0^{t-1}[n,\tau](x,s)dx.
\]
The key step in the derivation of \eqref{eq:transformation} uses the transformation of variables $y = \rho^{[t]}(x)$ and notice that $dx = dy \prod_{i=0}^{t-1} \tau'(\tau^{[i]}(y))$.

Next, we show that for any continuous and bounded test function $\phi:\mathcal{X} \rightarrow \mathbb{R}$,
\[
  \lim_{t\rightarrow\infty} \int_{\mathcal{X}_k} \phi(x)f_{k,t}(x) dx = \phi(s) .
\]
which would imply that $f_{k,t}$ converges to the Dirac delta function centered at $s$. By equation \eqref{eq:convK}, and employing the same transformation of variables as above,
\begin{align}
\lim_{t \rightarrow \infty} \int_{\mathcal{X}_k} \phi(x) f_{k,t}(x) dx &= \lim_{t \rightarrow \infty} \int_{\mathcal{X}_k} \phi(x) \frac{ f_{0}(\rho^{[t]} (x)) G_t(x,s) }{\int_{\mathcal{X}_k} f_{0}(\rho^{[t]} (z)) G_t(z,s) dz} dx \nonumber \\
& =  \lim_{t \rightarrow \infty} \int_{\mathcal{X}_k} \phi(\tau^{[t]}(y)) \frac{f_{0}(y) Q_0^{t-1}[n,\tau](y,s) } {\int_{\mathcal{X}_k} f_{0}(z) Q_0^{t-1}[n,\tau](z,s) dz} dy \nonumber \\
& =  \int_{\mathcal{X}_k}  \phi(s) \frac{f_{0}(y) \lim_{t \rightarrow \infty} Q_0^{t-1}[n,\tau](y,s) } {\int_{\mathcal{X}_k} f_{0}(z) \lim_{t \rightarrow \infty} Q_0^{t-1}[n,\tau](z,s) dz} dy
= \phi(s) . \nonumber
\end{align}
Pulling the limit into the integral follows Lemma \ref{lem:int}.

Finally, from equation \eqref{eq:conK2} and the same transformation of variables again, we have
\begin{align*}
\lim_{t \to \infty} \left(\frac{1}{n(s)}\right)^t \prod_{i=0}^{t} \mathbb{E}_{k,i}[n] &= \lim_{t \to \infty} \left(\frac{1}{\tau'(s)}\right)^t \int_{\mathcal{X}_k} n(x) G_t(x,s) f_0(\rho^{[t]}(x)) dx \\
&= \lim_{t \to \infty} \int_{\mathcal{X}_k} n(\tau^{[t]}(y)) Q_0^{t-1}[n,\tau](y,s) f_0(y) dy
\end{align*}
which is positive and finite, thus proving \eqref{eq:prodlim2}.
\end{proof}

We are finally in a position to prove Theorem \ref{thm}.
\begin{proof}[Proof of Theorem \ref{thm}]
Since the primitives are well-behaved and nice, we know by Proposition \ref{prop:parallel}, that for each $k \in \{0,\dots,K\}$, $F_{k,t}$ converges to a steady state $F_k^*$. We then employ Proposition \ref{prop:recover} to recover the steady state of the population process \eqref{main_nonlin} from the steady states of the parallel processes.

(i) Suppose $s_k^*$ is a source. We start by showing that niceness and genericity imply that $n(\cdot)/\tau'(\cdot)$ is endpoint maximal at $s_k^*$ in $\mathcal{X}_{k^*-1}$ when it is non-empty and in $\mathcal{X}_{k^*}$ when it is non-empty. Suppose that $\mathcal{X}_{k^*-1}$ is not empty. If $n(\cdot)$ were endpoint maximal at $s_{k^*-1}$ then we would have 
\[
\lim_{x \to s_{k^*-1}}n(x) > n(s_{k^*}) > \frac{n(s_{k^*})}{\tau'(s_{k^*})},
\]
where the final inequality follows from $k^*$ being a source. But this violates genericity. By niceness, we therefore must have that $n(\cdot)/\tau'(\cdot)$ is endpoint maximal at $s_k^*$ in $\mathcal{X}_{k^*-1}$. An analogous argument applies in the case in which $\mathcal{X}_{k^*}$ is not empty.

Since the primitives are nice in each interval $\mathcal{X}_k$, by Proposition \ref{prop:parallel} we have that $\lim_{t \to \infty} \mathbb{E}_{k,t}[n] = \mathbb{E}_k^*[n]$ for each $k\in\{0,\dots,K\}$. Each such $\mathbb{E}_k^*[n]$ is equal either to $n(\cdot)$ evaluated at a sink in $\mathcal{X}_k$ or to $n(\cdot)/\tau'(\cdot)$ evaluated at a source in $\mathcal{X}_k$. In particular, our argument above implies that for $k \in \{k^*-1,k^*\}$, $\mathbb{E}_{k}^*[n] =n (s_{k^*})/\tau'(s_{k^*}) $ but, by genericity, we have that for any $k \not \in \{k^*-1,k^*\}$, $\mathbb{E}_k^*[n]$ is strictly less than $n(s_{k^*})/\tau'(s_{k^*})$. From this, equation \eqref{eq:w}, and equations \eqref{eq:prodlim1} and \eqref{eq:prodlim2} in Proposition \ref{prop:parallel} we conclude that $\lim_{t\to \infty} w_{k,t}$ exists and is a positive constant for $k \in \{k^*-1,k^*\}$ for which $\mathcal{X}_k \neq \emptyset$ but it is zero for any $k \not\in\{k^*-1,k^*\}$. From \eqref{eq:recover} and Proposition \ref{prop:parallel}, we conclude that $\lim_{t\to \infty} f_t(x) \propto \lim_{t \to \infty} f_{k,t}(x) \propto \lim_{t \to \infty} G_t(x,s_{k^*})$ for $k \in \{k^*-1,k^*\}$ for which $\mathcal{X}_k \neq \emptyset$, and $\lim_{t\to \infty} f_t(x) =0$ for $k \not\in \{k^*-1,k^*\}$.

(ii) An analogous argument applies when $k^*$ is a sink.
\end{proof}

\subsection{Proofs for Section \ref{sec:comp_stats}}

\begin{lemma}\label{lem:comp_stat}
Suppose $\hat{F}_{t}$ and $F_{t}$ admit the densities $\hat{f}_{t}$ and $f_{t}$, respectively. If $\hat{f}_{t} \succsim f_{t}$ and $\hat{n} \succsim n$ then $\hat{f}_{t+1} \succsim f_{t+1}$, and if either ordering in the antecedent is strict then the consequent is also strict.
\end{lemma}
\begin{proof}[Proof of Lemma \ref{lem:comp_stat}]
By \eqref{eq:evolution}, for each $x \in \mathcal X$, we have
\[
\frac{ \hat{f}_{t+1}(x) }{ f_{t+1}(x) } \propto \frac{ \hat{n}(\rho(x)) }{ n(\rho(x)) } \frac{ \hat{f}_{t}(\rho(x)) }{ f_{t}(\rho(x)) } .
\]
Since each fraction on the right-hand side is non-decreasing in $x$, so is the one on the left-hand side.
\end{proof}

\begin{proof}[Proof of Proposition \ref{prop:comp1}]
The fact that $\hat{F}_{t}$ first-order stochastically dominates $F_{t}$ for all $t \geq 0$ follows from Lemma \ref{lem:comp_stat} and from the fact that MLR-dominance implies first-order stochastic dominance. If, additionally, all the conditions of Theorem \ref{thm} (i) hold then we directly have existence of atomless steady states to which the processes converge. Note however that, because the primitives $n$ and $\tau$ are identical, the steady states are also identical.
\end{proof}

\begin{proof}[Proof of Proposition \ref{prop:comp2}]
The fact that $\hat{F}_{t}$ first-order stochastically dominates $F_{t}$ for all $t \geq 0$ follows from Lemma \ref{lem:comp_stat} and from the fact that MLR-dominance implies first-order stochastic dominance. If, additionally, all the conditions of Theorem \ref{thm} (i) hold then we directly have existence of atomless steady states to which the processes converge. Moreover, since the fertility functions are distinct, the steady state limiting distributions are also distinct, which now implies that $\hat{f}^* \succ f^*$.
\end{proof}

\subsection{Proofs for Section \ref{sec:endo}}
\label{proofs:endogenous}

\subsubsection{Solution to the optimization problem}
We solve the optimization problem for a parent in generation $t$:
\begin{equation*}
\max_{n\geq0,c\geq0,e\geq0} \;  c +  \ln(n z')  \text{ subject to } c+ n e \bar{z}_t + z\phi n  \leq z \text{ and } z' = \theta_t e^\alpha z^\beta \bar{z}_t^\gamma  .
\end{equation*}
The budget constraint is not convex because of the term $ne$. We get around this problem by using the substitution $y = ne$ to obtain
\begin{equation*}
\max_{n\geq0,c\geq0,y\geq0} \;  c + (1-\alpha)\ln(n) + \alpha \ln(y) + \ln( \theta_t  z^\beta \bar{z}_t^\gamma)  \text{ subject to } c+ y \bar{z}_t + z\phi n  \leq z  .
\end{equation*}
Since $\alpha \in(0,1)$, this is now a convex optimization problem. Since the objective is strictly increasing in its arguments, the budget constraint binds. Furthermore we must have $n>0$ and $y>0$ at the optimum (since the marginal utility goes to infinity if either of these variables goes to zero). So, provided $c>0$, the problem reduces to
\begin{equation*}
\max_{n,y} \;  z - y \bar{z}_t - z\phi n + (1-\alpha)\ln(n) + \alpha \ln(y) + \ln( \theta_t  z^\beta \bar{z}_t^\gamma) ,
\end{equation*}
which has the solution $y = \frac{\alpha}{ \bar{z}_t}$ and $n = \frac{1-\alpha}{ \phi z}$ and, in turn, this implies $e = \frac{\alpha}{1-\alpha} \frac{\phi z}{ \bar{z}_t}$ and $c =  z - 1$. Clearly $c>0$ when $z>1$. And, when $c=0$ the problem is
\begin{equation*}
\max_{n,y} \;  (1-\alpha)\ln(n) + \alpha \ln(y) + \ln( \theta_t  z^\beta \bar{z}_t^\gamma)  \text{ subject to } y \bar{z}_t + z\phi n  = z  ,
\end{equation*}
which has solution $y = \alpha \frac{z}{\bar{z}_t}$ and $n = \frac{1-\alpha}{\phi}$ and, in turn, this implies $e = \frac{\alpha}{1-\alpha} \frac{\phi z}{ \bar{z}_t}$. This therefore recovers the optimum given in the main text.

\subsubsection{Growth}
Our primary interest was to showcase the applicability of the results of Section \ref{sec:results} to a model of growth with endogenous differential fertility, and we did this in Section \ref{sec:distro}. However, for the sake of completing the analysis of the model, we now focus on averages (specifically, average capital and fertility) rather than on distributional outcomes. To this end, define \defn{normalized average capital} to be:
\[
\chi_t := \frac{\bar{z}_t}{(1+\varrho)^t} .
\]
Assume $\alpha + \beta > 1$ and that the economy in generation $t$ is at the Pareto steady state; i.e. the distribution of relative capital in generation $t$ is Pareto$[1, 1/(\alpha+\beta - 1)]$. Then, by a simple transformation of variables, the distribution of capital itself in generation $t$ is Pareto$[\underline{z}_t, 1/(\alpha+\beta - 1)]$. Moreover, average capital evolves according to  $\bar{z}_{t+1} = (2 - (\alpha + \beta))^{\alpha + \beta - 1} \zeta_t  \bar{z}_t^{\beta + \gamma}$.\footnote{A Pareto$[c,\nu]$ distribution has mean $\frac{\nu c}{\nu - 1}$ for $\nu > 1$, so here average capital is $\bar{z}_t = \underline{z}_t/(2 - (\alpha+\beta))$, and from \eqref{eq:priv_trans} we have that $\underline{z}_{t+1} = \zeta_t \underline{z}_t^{\alpha + \beta} \bar{z}_t^{\gamma - \alpha}$. Combining these gives us the expression in the text.} Expressing this in terms of normalized average capital yields\footnote{ Recall, $\zeta_t := \theta_t \left( \frac{\alpha}{1-\alpha} \phi  \right)^\alpha$ and $\theta_t := \theta (1+\varrho)^{(1-\beta-\gamma)t}$, so $\zeta_t = \zeta_0 (1+\varrho)^{(1-\beta-\gamma)t}$.}
 \begin{equation}\label{evo:priv1}
 \chi_{t+1} = \frac{\zeta_0}{1+\varrho} (2 - (\alpha + \beta))^{\alpha + \beta - 1} \chi_t^{\beta + \gamma} .
 \end{equation}
 On the other hand if $\alpha + \beta < 1$ and the distribution of relative capital is degenerate at 1 then, from \eqref{eq:priv_trans}, capital per capita evolves according to $ \bar{z}_{t+1} = \zeta_t  \bar{z}_t^{\beta + \gamma}$ which, expressed in terms of normalized average capital, gives us
 \begin{equation}\label{evo:priv2}
 \chi_{t+1} = \frac{\zeta_0}{1+\varrho} \chi_t^{\beta + \gamma} .
 \end{equation}
Equations \eqref{evo:priv1} and \eqref{evo:priv2}, respectively, tell us how capital per capita evolves in the Pareto steady state and the degenerate steady state. (i) Under exogenous growth ($\beta + \gamma<1$), each of \eqref{evo:priv1} and \eqref{evo:priv2} has a unique positive fixed point. At the fixed point of either \eqref{evo:priv1} or \eqref{evo:priv2}, un-normalized average capital $\bar{z}_t$ grows at the fixed rate $\varrho$. In other words, average captial grows at the same rate in the balanced growth path of the Pareto and of the degenerate steady state. However, the fixed point of \eqref{evo:priv1} has a smaller value than that of \eqref{evo:priv2} implying that the \emph{level} of average capital is lower in the Pareto steady state. (ii) Under endogenous growth ($\beta + \gamma = 1$),  \eqref{evo:priv1} and \eqref{evo:priv2} are linear processes which implies that the steady state growth rates themselves can be ranked, with the average capital growth being higher in the degenerate steady state than in the Pareto steady state.

Finally, we turn to population growth. An interesting feature of our model is that endogenous downward-sloping fertility is also accompanied by a reduction in mean fertility as capital per capita grows in steady state. Integrating the expression for $\hat{n}_t$ by its corresponding steady state distribution we get
\[
  \mathbb{E}_t[\hat{n}_t] := \int_\mathcal{X} \hat{n}_t(x)d \hat{F}_t(x) = 
  \begin{cases}
  \frac{1-\alpha}{\underline{z}_t \phi} \frac{1}{\alpha + \beta} & \text{ in the Pareto steady state} \\
  \frac{1-\alpha}{\underline{z}_t \phi} & \text{ in the degenerate steady state}
  \end{cases}
\]
It is easily verified that, in any of our steady states, the lowest level of capital $\underline{z}_t$ must grow at the same rate as average capital $\bar{z}_t$. The above therefore shows that, in any steady state with a positive growth rate of capital, average fertility (i.e. the population growth factor) must be declining at the rate at which capital grows. The steady states of our model thus capture both stylized facts about the fertility income relationship: fertility is lower at the more affluent cross-section of the population, and societies in aggregate reduce their fertility as they grow richer.\footnote{\cite*{brunnschweiler2021wealth} highlight the unsatisfactory nature of a constant long-run population growth rate, which is present in many existing theories of growth. We do not have constant long-run population growth in our model: in our case, technology-driven growth \emph{leads} to a reduction in the population growth factor in steady state. In the spirit of \cite*{brunnschweiler2021wealth}, one could even imagine adding a feedback mechanism to our model so that growth slows down as population declines, which would lead to the population \emph{level} eventually stabilizing.}

\section{Branching systems}\label{sec:branching}

We discuss how our model relates to the operator-theoretic approach to deterministic population dynamics; e.g. see \citet*{jagers2001deterministic,lasota2008probabilistic}. Consider an operator on the space of Lebesgue integrable functions $P : \mathcal{L}^1 \to \mathcal{L}^1$ of the form
\[
(Pf)(x) = \int K(x,z) f(z) dz .
\]
When $K$ is a stochastic kernel and $f$ is a density, $P$ is a Markov operator. More general kernels arise in structured population and branching models, where $K$ may encode expected reproduction rather than probabilities. Related work includes the Frobenius–Perron theory of deterministic transformations \citep{lasota1973,lasota2008probabilistic} and the theory of deterministic or expected evolution in general branching populations \citep{jagers2001deterministic}.

To see the link with our paper, if $F_t$ has density $f$, then  \eqref{main_nonlin} can be expressed as a mapping $P$ that is defined by
\[
 (Pf)(x) = \frac{1}{\int n(y) f(y) dy} \int K(x,z) f(z) dz ,
\]
where $K(x,z) = n(z) \, \delta( x - \tau(z))$ and $\delta$ is the Dirac distribution. The normalization by $\int n(y)f(y)dy$ makes this map nonlinear in $f$, but one can in principle analyze an un-normalized version---which would be linear---and then re-normalize appropriately.

Under suitable compactness or positivity assumptions on a linear operator $P$, one can obtain eigenfunctions, and after normalization these may correspond to invariant distributions. Such results typically rely on $P$ satisfying certain smoothing or expansion properties that are absent in our model. Indeed, evaluating the integral under our kernel gives us the counterpart to \eqref{eq:evolution} from the main text, namely,
\[
(Pf)(x) = \frac{1}{\int n(y) f(y) dy} \, \frac{n(\rho(x))}{\tau'(\rho(x))} f(\rho(x))
\]
which simply takes each point in $f$ and maps it to a new point on a distribution $Pf$. Our mapping $P$ does not smooth out oscillations in $f$ but rather preserves them. For such (possibly non-compact) maps, \cite{lasota1973} prove existence of absolutely continuous invariant measures under expansion conditions, and \cite{lasota2008probabilistic} give stronger asymptotic stability results for several classes of expanding transformations. But expansion conditions in our model would roughly correspond to $\tau'(x) > 1$ for all $x$, which is too restrictive in economic settings. \cite{jagers2001deterministic} also studies general deterministic branching processes but the Perron–Frobenius and asymptotic results there rely on communication and recurrence hypotheses for positive reproduction kernels that are generally not met by our deterministic Dirac kernel. More broadly, the studies above on deterministic branching processes tend to analyze, at a high-level, conditions on operators that yield invariant distributions or asymptotic stability. The sufficient conditions can therefore be stronger than they may need to be in specific settings. By committing to a specific way in which fertility and transmission interact, our model operates at a lower level but it also allows for the kernel to have properties that, while natural in economic settings, do not meet the conditions of \cite{lasota1973}, \cite{lasota2008probabilistic}, or \cite{jagers2001deterministic}.


\begin{thebibliography}{}

\bibitem[\protect\citeauthoryear{Atkinson}{Atkinson}{1980}]{atkinson80}
Atkinson, A.~B. (1980).
\newblock Inheritance and the redistribution of wealth.
\newblock In G.~Hughes and G.~Heal (Eds.), {\em Public policy and the tax
  system}, pp.\  36--66. Allen and Unwin.

\bibitem[\protect\citeauthoryear{Azariadis and Stachurski}{Azariadis and
  Stachurski}{2005}]{Azariadis2005Chapter}
Azariadis, C. and J.~Stachurski (2005).
\newblock Poverty traps.
\newblock In {\em Handbook of Economic Growth}, Volume~1, pp.\  295--384.
  Elsevier.

\bibitem[\protect\citeauthoryear{Balbo, Billari, and Mills}{Balbo
  et~al.}{2013}]{balbo13}
Balbo, N., F.~C. Billari, and M.~Mills (2013).
\newblock Fertility in advanced societies: a review of research.
\newblock {\em European Journal of Population/Revue Europ{\'e}enne de
  D{\'e}mographie\/}~{\em 29\/}(1), 1--38.

\bibitem[\protect\citeauthoryear{Barro and Becker}{Barro and
  Becker}{1989}]{barro89}
Barro, R.~J. and G.~S. Becker (1989).
\newblock Fertility choice in a model of economic growth.
\newblock {\em Econometrica\/}~{\em 57\/}(2), 481--501.

\bibitem[\protect\citeauthoryear{Beare and Toda}{Beare and
  Toda}{2022}]{beare2022determination}
Beare, B.~K. and A.~A. Toda (2022).
\newblock Determination of pareto exponents in economic models driven by markov
  multiplicative processes.
\newblock {\em Econometrica\/}~{\em 90\/}(4), 1811--1833.

\bibitem[\protect\citeauthoryear{Becker and Tomes}{Becker and
  Tomes}{1979}]{becker79}
Becker, G.~S. and N.~Tomes (1979).
\newblock An equilibrium theory of the distribution of income and
  intergenerational mobility.
\newblock {\em Journal of Political Economy\/}~{\em 87\/}(6), 1153--1189.

\bibitem[\protect\citeauthoryear{Benhabib and Bisin}{Benhabib and
  Bisin}{2018}]{benhabib16}
Benhabib, J. and A.~Bisin (2018).
\newblock Skewed wealth distributions: theory and empirics.
\newblock {\em Journal of Economic Literature\/}~{\em 56\/}(4), 1261--91.

\bibitem[\protect\citeauthoryear{Benhabib, Bisin, and Zhu}{Benhabib
  et~al.}{2011}]{benhabib11}
Benhabib, J., A.~Bisin, and S.~Zhu (2011).
\newblock The distribution of wealth and fiscal policy in economies with
  finitely lived agents.
\newblock {\em Econometrica\/}~{\em 79\/}(1), 123--157.

\bibitem[\protect\citeauthoryear{Benhabib, Bisin, and Zhu}{Benhabib
  et~al.}{2015}]{benhabib14}
Benhabib, J., A.~Bisin, and S.~Zhu (2015).
\newblock The wealth distribution in bewley economies with capital income risk.
\newblock {\em Journal of Economic Theory\/}~{\em 159}, 489--515.

\bibitem[\protect\citeauthoryear{Blinder}{Blinder}{1973}]{blinder73}
Blinder, A.~S. (1973).
\newblock A model of inherited wealth.
\newblock {\em Quarterly Journal of Economics\/}~{\em 87\/}(4), 608--626.

\bibitem[\protect\citeauthoryear{Bosi, Boucekkine, and Seegmuller}{Bosi
  et~al.}{2011}]{bosi11}
Bosi, S., R.~Boucekkine, and T.~Seegmuller (2011).
\newblock The dynamics of wealth inequality under endogenous fertility: A
  remark on the barro-becker model with heterogenous endowments.
\newblock {\em Theoretical Economics Letters\/}~{\em 1\/}(1), 3--7.

\bibitem[\protect\citeauthoryear{Brunnschweiler, Peretto, and
  Valente}{Brunnschweiler et~al.}{2021}]{brunnschweiler2021wealth}
Brunnschweiler, C.~N., P.~F. Peretto, and S.~Valente (2021).
\newblock Wealth creation, wealth dilution and demography.
\newblock {\em Journal of Monetary Economics\/}~{\em 117}, 441--459.

\bibitem[\protect\citeauthoryear{Bruze}{Bruze}{2015}]{bruze15}
Bruze, G. (2015).
\newblock Male and female marriage returns to schooling.
\newblock {\em International Economic Review\/}~{\em 56\/}(1), 207--234.

\bibitem[\protect\citeauthoryear{Cavalcanti, Kocharkov, and Santos}{Cavalcanti
  et~al.}{2021}]{cavalcanti2021family}
Cavalcanti, T., G.~Kocharkov, and C.~Santos (2021).
\newblock Family planning and development: Aggregate effects of contraceptive
  use.
\newblock {\em Economic Journal\/}~{\em 131\/}(634), 624--657.

\bibitem[\protect\citeauthoryear{Cavalcanti and Giannitsarou}{Cavalcanti and
  Giannitsarou}{2017}]{cavalcanti2017growth}
Cavalcanti, T.~V. and C.~Giannitsarou (2017).
\newblock Growth and human capital: a network approach.
\newblock {\em Economic Journal\/}~{\em 127\/}(603), 1279--1317.

\bibitem[\protect\citeauthoryear{Champernowne}{Champernowne}{1953}]{champernowne1953model}
Champernowne, D.~G. (1953).
\newblock A model of income distribution.
\newblock {\em The Economic Journal\/}~{\em 63\/}(250), 318--351.

\bibitem[\protect\citeauthoryear{Chu and Koo}{Chu and Koo}{1990}]{chu90}
Chu, C.~C. and H.-W. Koo (1990).
\newblock Intergenerational income-group mobility and differential fertility.
\newblock {\em American Economic Review\/}~{\em 80\/}(5), 1125--1138.

\bibitem[\protect\citeauthoryear{C{\'o}rdoba, Liu, and Ripoll}{C{\'o}rdoba
  et~al.}{2016}]{cordoba16}
C{\'o}rdoba, J.~C., X.~Liu, and M.~Ripoll (2016).
\newblock Fertility, social mobility and long run inequality.
\newblock {\em Journal of Monetary Economics\/}~{\em 77}, 103--124.

\bibitem[\protect\citeauthoryear{Dahan and Tsiddon}{Dahan and
  Tsiddon}{1998}]{dahan98}
Dahan, M. and D.~Tsiddon (1998).
\newblock Demographic transition, income distribution, and economic growth.
\newblock {\em Journal of Economic Growth\/}~{\em 3\/}(1), 29--52.

\bibitem[\protect\citeauthoryear{Davies and Shorrocks}{Davies and
  Shorrocks}{2000}]{davies00}
Davies, J.~B. and A.~F. Shorrocks (2000).
\newblock The distribution of wealth.
\newblock In A.~B. Atkinson and F.~Bourguignon (Eds.), {\em Handbook of Income
  Distribution}, Volume~1, pp.\  605--675. Elsevier.

\bibitem[\protect\citeauthoryear{De~La~Croix and Doepke}{De~La~Croix and
  Doepke}{2003}]{delacroix03}
De~La~Croix, D. and M.~Doepke (2003).
\newblock Inequality and growth: why differential fertility matters.
\newblock {\em American Economic Review\/}~{\em 93\/}(4), 1091--1113.

\bibitem[\protect\citeauthoryear{De~La~Croix and Doepke}{De~La~Croix and
  Doepke}{2004}]{Croix2004Public}
De~La~Croix, D. and M.~Doepke (2004).
\newblock Public versus private education when differential fertility matters.
\newblock {\em Journal of Development Economics\/}~{\em 73\/}(2), 607--629.

\bibitem[\protect\citeauthoryear{Dribe, Oris, and Pozzi}{Dribe
  et~al.}{2014}]{dribe14}
Dribe, M., M.~Oris, and L.~Pozzi (2014).
\newblock Socioeconomic status and fertility before, during, and after the
  demographic transition: An introduction.
\newblock {\em Demographic Research\/}~{\em 31}, 161--182.

\bibitem[\protect\citeauthoryear{Fan and Zhang}{Fan and
  Zhang}{2013}]{fan2013differential}
Fan, C.~S. and J.~Zhang (2013).
\newblock Differential fertility and intergenerational mobility under private
  versus public education.
\newblock {\em Journal of Population Economics\/}~{\em 26\/}(3), 907--941.

\bibitem[\protect\citeauthoryear{Gabaix}{Gabaix}{2009}]{gabaix09}
Gabaix, X. (2009).
\newblock Power laws in economics and finance.
\newblock {\em Annual Review of Economics\/}~{\em 1}, 244--294.

\bibitem[\protect\citeauthoryear{Gabaix, Lasry, Lions, and Moll}{Gabaix
  et~al.}{2016}]{gabaix15}
Gabaix, X., J.-M. Lasry, P.-L. Lions, and B.~Moll (2016).
\newblock The dynamics of inequality.
\newblock {\em Econometrica\/}~{\em 84\/}(6), 2071--2111.

\bibitem[\protect\citeauthoryear{Galor and Zeira}{Galor and
  Zeira}{1993}]{Galor1993Income}
Galor, O. and J.~Zeira (1993).
\newblock Income distribution and macroeconomics.
\newblock {\em Review of Economic Studies\/}~{\em 60\/}(1), 35--52.

\bibitem[\protect\citeauthoryear{Glomm and Ravikumar}{Glomm and
  Ravikumar}{1992}]{Glomm1992Public}
Glomm, G. and B.~Ravikumar (1992).
\newblock Public versus private investment in human capital: Endogenous growth
  and income inequality.
\newblock {\em Journal of Political Economy\/}~{\em 100\/}(4), 818--834.

\bibitem[\protect\citeauthoryear{Greenwood, Guner, Kocharkov, and
  Santos}{Greenwood et~al.}{2016}]{greenwood16}
Greenwood, J., N.~Guner, G.~Kocharkov, and C.~Santos (2016).
\newblock Technology and the changing family: A unified model of marriage,
  divorce, educational attainment, and married female labor-force
  participation.
\newblock {\em American Economic Journal: Macroeconomics\/}~{\em 8\/}(1),
  1--41.

\bibitem[\protect\citeauthoryear{Jagers}{Jagers}{2001}]{jagers2001deterministic}
Jagers, P. (2001).
\newblock The deterministic evolution of general branching populations.
\newblock {\em State of the Art in Probability and Statistics\/}~{\em 36},
  384--398.

\bibitem[\protect\citeauthoryear{Jones}{Jones}{2015}]{jones15}
Jones, C.~I. (2015).
\newblock Pareto and piketty: the macroeconomics of top income and wealth
  inequality.
\newblock {\em Journal of Economic Perspectives\/}~{\em 29\/}(1), 29--46.

\bibitem[\protect\citeauthoryear{Jones, Schoonbroodt, and Tertilt}{Jones
  et~al.}{2010}]{jones08}
Jones, L.~E., A.~Schoonbroodt, and M.~Tertilt (2010).
\newblock Fertility theories: can they explain the negative fertility-income
  relationship?
\newblock In {\em Demography and the Economy}, pp.\  43--100. University of
  Chicago Press.

\bibitem[\protect\citeauthoryear{Kremer and Chen}{Kremer and
  Chen}{2002}]{kremer02}
Kremer, M. and D.~L. Chen (2002).
\newblock Income distribution dynamics with endogenous fertility.
\newblock {\em Journal of Economic Growth\/}~{\em 7\/}(3), 227--258.

\bibitem[\protect\citeauthoryear{Lam}{Lam}{1986}]{lam86}
Lam, D. (1986).
\newblock The dynamics of population growth, differential fertility, and
  inequality.
\newblock {\em American Economic Review\/}~{\em 76\/}(5), 1103--1116.

\bibitem[\protect\citeauthoryear{Lasota and Mackey}{Lasota and
  Mackey}{2008}]{lasota2008probabilistic}
Lasota, A. and M.~C. Mackey (2008).
\newblock {\em Probabilistic properties of deterministic systems}.
\newblock Cambridge University Press.

\bibitem[\protect\citeauthoryear{Lasota and Yorke}{Lasota and
  Yorke}{1973}]{lasota1973}
Lasota, A. and J.~A. Yorke (1973).
\newblock On the existence of invariant measures for piecewise monotonic
  transformations.
\newblock {\em Transactions of the American Mathematical Society\/}~{\em 186},
  481--488.

\bibitem[\protect\citeauthoryear{Loury}{Loury}{1981}]{loury81}
Loury, G.~C. (1981).
\newblock Intergenerational transfers and the distribution of earnings.
\newblock {\em Econometrica\/}~{\em 49\/}(4), 843--867.

\bibitem[\protect\citeauthoryear{Mace}{Mace}{2008}]{mace08}
Mace, R. (2008).
\newblock Reproducing in cities.
\newblock {\em Science\/}~{\em 319\/}(5864), 764--766.

\bibitem[\protect\citeauthoryear{Mare}{Mare}{2016}]{mare16}
Mare, R.~D. (2016).
\newblock Educational homogamy in two gilded ages evidence from
  inter-generational social mobility data.
\newblock {\em Annals of the American Academy of Political and Social
  Science\/}~{\em 663\/}(1), 117--139.

\bibitem[\protect\citeauthoryear{Moav}{Moav}{2005}]{moav05}
Moav, O. (2005).
\newblock Cheap children and the persistence of poverty.
\newblock {\em Economic Journal\/}~{\em 115\/}(500), 88--110.

\bibitem[\protect\citeauthoryear{Morand}{Morand}{1999}]{morand99}
Morand, O.~F. (1999).
\newblock Endogenous fertility, income distribution, and growth.
\newblock {\em Journal of Economic Growth\/}~{\em 4\/}(3), 331--349.

\bibitem[\protect\citeauthoryear{Piketty and Zucman}{Piketty and
  Zucman}{2015}]{piketty15}
Piketty, T. and G.~Zucman (2015).
\newblock Wealth and inheritance in the long run.
\newblock In A.~B. Atkinson and F.~Bourguignon (Eds.), {\em Handbook of Income
  Distribution}, Volume~2, pp.\  1303--1368. Elsevier.

\bibitem[\protect\citeauthoryear{Rangazas}{Rangazas}{2000}]{rangazas2000schooling}
Rangazas, P. (2000).
\newblock Schooling and economic growth: A king--rebelo experiment with human
  capital.
\newblock {\em Journal of Monetary Economics\/}~{\em 46\/}(2), 397--416.

\bibitem[\protect\citeauthoryear{Skirbekk}{Skirbekk}{2008}]{skirbekk08}
Skirbekk, V. (2008).
\newblock Fertility trends by social status.
\newblock {\em Demographic Research\/}~{\em 18\/}(5), 145--180.

\bibitem[\protect\citeauthoryear{Stiglitz}{Stiglitz}{1969}]{stiglitz69}
Stiglitz, J.~E. (1969).
\newblock Distribution of income and wealth among individuals.
\newblock {\em Econometrica\/}~{\em 37\/}(3), 382--397.

\bibitem[\protect\citeauthoryear{Straub and Wenig}{Straub and
  Wenig}{1984}]{straub84}
Straub, M. and A.~Wenig (1984).
\newblock Human fertility and the distribution of wealth.
\newblock In G.~Steinmann (Ed.), {\em Economic Consequences of Population
  Change in Industrialized Countries: Proceedings of the Conference on
  Population Economics Held at the University of Paderborn, West Germany, June
  1--3, 1983}, pp.\  68--86. Springer.

\bibitem[\protect\citeauthoryear{Vaughan}{Vaughan}{1979}]{vaughan79}
Vaughan, R. (1979).
\newblock Class behaviour and the distribution of wealth.
\newblock {\em Review of Economic Studies\/}~{\em 46\/}(3), 447--465.

\bibitem[\protect\citeauthoryear{Vogl}{Vogl}{2016}]{vogl16b}
Vogl, T.~S. (2016).
\newblock Differential fertility, human capital, and development.
\newblock {\em Review of Economic Studies\/}~{\em 83\/}(1), 365--401.

\bibitem[\protect\citeauthoryear{Wold and Whittle}{Wold and
  Whittle}{1957}]{wold1957model}
Wold, H.~O. and P.~Whittle (1957).
\newblock A model explaining the pareto distribution of wealth.
\newblock {\em Econometrica\/}, 591--595.

\bibitem[\protect\citeauthoryear{Yule}{Yule}{1925}]{yule1925ii}
Yule, G.~U. (1925).
\newblock Ii.—a mathematical theory of evolution, based on the conclusions of
  dr. jc willis, fr s.
\newblock {\em Philosophical transactions of the Royal Society of London Series
  B\/}~{\em 213\/}(402-410), 21--87.

\end{thebibliography}
\end{document}